\documentclass[11pt]{article}

\usepackage{lmodern} 
\usepackage[utf8]{inputenc}
\usepackage[T1]{fontenc}

\usepackage[numbers]{natbib}
\usepackage{amsmath}
\usepackage{amsfonts,amsthm}  
\usepackage{fullpage}
\usepackage{cleveref} 
\usepackage{graphicx, subfigure}
\usepackage{algorithm, algorithmic}
\usepackage{mdframed}
\usepackage[inline]{enumitem} 
\usepackage{thmtools,thm-restate} 
\usepackage{url}

\newtheorem{theorem}{Theorem}
\newtheorem{lemma}[theorem]{Lemma}
\newtheorem{definition}[theorem]{Definition}

\newcommand{\eps}{\epsilon}

\newcommand{\mcT}{\mathcal{T}} 
\newcommand{\mcS}{\mathcal{S}} 
\newcommand{\mcG}{\mathcal{G}} 
\newcommand{\WCT}{\mathcal{WCT}} 

\newcommand{\poly}{\mathrm{poly}} 
\newcommand{\NC}{\text{NC}} 
    {
      \begin{enumerate}[label=\emph{\roman*)}]%
    }
    {%
    \end{enumerate}%
    }
\newenvironment{enumRomanHor}
    {%
      \begin{enumerate*}[label=\emph{\roman*)}]%
    }
    {%
    \end{enumerate*}%
    }

\newcommand{\E}{\mathop{{}\mathbb{E}}}

\usepackage{xcolor}

\newcommand{\FullOrShort}{full}
\ifthenelse{\equal{\FullOrShort}{full}}{
  \newcommand{\fullOnly}[1]{#1}
  \newcommand{\shortOnly}[1]{}
}{

  \newcommand{\fullOnly}[1]{}
  \newcommand{\shortOnly}[1]{#1}
}

\begin{document}

\title{Broadcasting in Noisy Radio Networks}

\newcommand{\ns}{\normalsize}

\author{\ns Keren Censor-Hillel\footnotemark[1]\\ \ns Technion\\ \ns \texttt{ckeren@cs.technion.ac.il} \and
                                \ns Bernhard Haeupler\footnotemark[2]\\ \ns Carnegie Mellon University\\ \ns \texttt{haeupler@cs.cmu.edu} \and
                                \ns D. Ellis Hershkowitz\footnotemark[2]\\ \ns Carnegie Mellon University\\ \ns \texttt{dhershko@cs.cmu.edu} \and
                                \ns Goran Zuzic\footnotemark[2]\\ \ns Carnegie Mellon University\\ \ns \texttt{gzuzic@cs.cmu.edu}}

\date{}

\renewcommand{\thefootnote}{\fnsymbol{footnote}}

\footnotetext[1]{Supported in part by the Israel Science Foundation (grant 1696/14) and the Binational Science Foundation (grant 2015803).}
\footnotetext[2]{Supported in part by the National Science Foundation through grants CCF-1527110 and CCF-1618280.}

\renewcommand{\thefootnote}{\arabic{footnote}}
\setcounter{footnote}{0}

\maketitle

\begin{abstract}
The widely-studied radio network model [Chlamtac and Kutten, 1985] is a graph-based description that captures the inherent impact of collisions in wireless communication. In this model, the strong assumption is made that node $v$ receives a message from a neighbor if and only if exactly one of its neighbors broadcasts.

We relax this assumption by introducing a new \emph{noisy radio network model} in which random faults occur at senders or receivers. Specifically, for a constant noise parameter $p \in [0,1)$, either every sender has probability $p$ of transmitting noise or every receiver of a single transmission in its neighborhood has probability $p$ of receiving noise.

We first study \emph{single-message broadcast} algorithms in noisy radio networks and show that the Decay algorithm [Bar-Yehuda et al., 1992] remains robust in the noisy model while the diameter-linear algorithm of Gasieniec et al., 2007 does not. We give a modified version of the algorithm of Gasieniec et al., 2007 that is robust to sender and receiver faults, and extend both this modified algorithm and the Decay algorithm to robust \emph{multi-message broadcast} algorithms, broadcasting $\Omega\left(\frac{1}{\log n \log \log n}\right)$ and $\Omega\left(\frac{1}{\log n}\right)$ messages per round, respectively.

We next investigate the extent to which (network) coding improves throughput in noisy radio networks. In particular, we study the coding cap -- the ratio of the throughput of coding to that of routing -- in noisy radio networks. We address the previously perplexing result of Alon~et~al.~2014 that worst case coding throughput is no better than worst case routing throughput up to constants: we show that the worst case throughput performance of coding is, in fact, superior to that of routing -- by a $\Theta(\log(n))$ gap -- provided receiver faults are introduced. However, we show that sender faults have little effect on throughput. In particular, we show that any coding or routing scheme for the noiseless setting can be transformed to be robust to sender faults with only a constant throughput overhead. These transformations imply that the results of Alon et al., 2014 carry over to noisy radio networks with sender faults as well. As a result, if sender faults are introduced then there exist topologies for which there is a $\Theta(\log\log{n})$ gap, but the worst case throughput across all topologies is $\Theta\left(\frac{1}{\log{n}}\right)$ for both coding and routing.
\end{abstract}

\vfill
\vfill
\vfill

\setcounter{page}{0}
\thispagestyle{empty}

\newpage

\section{Introduction}
Broadcasting messages throughout a network is one the most important network communication primitives.
The classic radio network model of \citet{onBroadChlamtac} was designed as a mathematical model to study broadcasting in a wireless setting. In this model, a radio network is represented by an undirected graph of $n$ nodes that communicate by transmitting messages during synchronized rounds. During each round, every node can either listen to the radio channel or transmit a message. A listening node $v$ receives a message if and only if exactly one of its neighbors transmits during that round. If more than a single neighbor transmits in a certain round, a collision occurs and node $v$ does not receive any of the transmitted messages.

The broadcast task in the classic radio network model -- in which a single message or multiple messages need to be disseminated throughout the network -- has been studied extensively~\cite{bar1992time, alon1991lower, Czumaj2006115,Kowalski2005,Peleg2007,Yehuda1989}. Classic algorithms \cite{bar1992time,gkasieniec2007faster} typically employ \textit{routing}. That is, nodes only transmit one of the messages the algorithm is disseminating. Alternatively, \textit{(network) coding} approaches, in which nodes transmit multiple messages coded together, have been of recent interest~\cite{Ahlswede2006,haeuplersodap1843,Khabbazian2011,ghaffari2015randomized,Li2009,Halperin2003,Censor-Hillel2014}.

Despite having been extensively studied, a notable deficit of the classic radio network model is its inability to model random noise. Specifically, the classic model assumes that a message that is sent without collisions is correctly received. However, this assumption is overly optimistic for real environments, in which noise may impede communication.

\paragraph{Our contribution.}
In this paper, we introduce the \emph{noisy radio network model}, in which the classic graph-based model of \citet{onBroadChlamtac} is augmented with random faults. In particular, for a constant fault parameter $p \in [0,1)$, every transmission may be noisy with probability $p$ (\emph{sender fault}), or a node $v$ that would otherwise receive a message with probability $p$ receives noise instead (\emph{receiver fault}). The faults occur independently at each node.

We begin by studying the extent to which the performance of existing single-message broadcasting algorithms deteriorates: we show that while the Decay algorithm of Bar-Yehuda, Goldreich and Itai ~\cite{bar1992time} is robust to faults, the diameter-linear algorithm of G\k{a}sieniec, Peleg and Xin~\cite{gkasieniec2007faster} (which we call FASTBC) deteriorates considerably. We then develop a new single-message, diameter-linear algorithm for the noisy radio network model. Moreover, we describe how to extend both Decay and our modified algorithm to multi-message broadcasting algorithms, achieving throughputs of $\Omega\left(\frac{1}{\log n}\right)$ and $\Omega\left(\frac{1}{\log n \log\log n}\right)$ messages per round, respectively.

The main challenge that arises when designing fault-robust algorithms is avoiding careful deterministic round synchronization; that is, random faults prevent nodes from knowing exactly which other nodes have a message in any given round. Prior algorithms used round synchronization of this nature but can be made fault-robust by repeating certain otherwise fragile subroutines.

Additionally, we study the power of (network) coding in the new noisy radio network model. In particular, we examine the \emph{coding gap} in this model: roughly the ratio between the throughput of coding to the throughput of routing. For the classic radio network model, recent work of \citet{haeuplersodap1843} demonstrated an $\Omega(\log \log n)$ coding gap on certain networks. However, the same work shows that up to constant factors, in worst case topologies, coding performs no better than routing, when one broadcasts a very large number of messages. This runs contrary to the intuition that coding ought to improve the throughput of communication.

We resolve these counterintuitive results by showing that coding is, in fact, much more powerful than routing, \emph{provided receiver faults occur}. Not only do we show that with receiver faults coding throughput is superior to routing by a $\Theta(\log n)$ gap on certain topologies, but we prove that the worst case performance of coding is superior to that of routing by a $\Theta(\log n)$ gap. These results emphasize that in practical settings coding can, in fact, significantly improve broadcast efficiency.

Lastly, we show that any algorithm for the classical radio network model can be made robust to sender faults at the price of a constant factor in throughput. This implies that the counterintuitive results of \citet{haeuplersodap1843} carry over to the noisy radio network model \emph{with sender faults}.

\section{Related Work}
The radio network model was introduced more than 30 years ago in the pioneering work of \citet{onBroadChlamtac}. Since then computation in radio networks has been extensively studied. An excellent survey is given by~\citet{Peleg2007}. Below, we discuss the most related work.

\textbf{Single-message broadcasting in radio networks:} In~\citet{bar1992time}, it was shown that single-message broadcast in a known topology can be achieved in $O(D\log n + \log^2 n)$ rounds, where $D$ is the diameter of the network. This was improved by~\citet{gkasieniec2007faster} and ~\citet{Kowalski2007} who showed that if the topology is known broadcast can be completed in $O(D + \log^2 n)$ rounds. If the topology is not known, \citet{Czumaj2006115} show that $O(D \log (n/D) + \log^2 n)$ rounds suffice. By the $\Omega(\log^2{n})$ and $\Omega(D \log (n/D))$ lower bounds of~\citet{alon1991lower} and~\citet{Kushilevitz1993}, respectively, the above complexities are optimal. Lastly, if \emph{collision detection} is available -- a stronger assumption than known topology -- a $O(D + \poly\log n)$-round algorithm is achievable as shown by~\citet{ghaffari2015randomized}.

\textbf{Multi-message broadcasting in radio networks:} The earliest work for broadcasting $k$ messages is~\citet{Yehuda1989}, who gave an algorithm that used $O((n + (k+D)\log n)\log \Delta)$ rounds, where $\Delta$ is the maximum node degree. A deterministic algorithm that works in $O(n \log^4 n + k\log^3 n)$ rounds was given by~\citet{chlebus2011efficient}. The first algorithm to beat the performance of~\citet{Yehuda1989} and to achieve an optimal $O(k \log n)$ dependence on the number of messages was given by \citet{ghaffari2013fast}. This algorithm completes in $O(k\log n + D\log n/D + \poly\log n)$ rounds. The $\Omega(k \log n)$ lower bound for the number of rounds is given by \citet{haeuplersodap1843}. \citet{ghaffari2015randomized} provided an algorithm that completes in $O(k\log n + D + \log^2 n)$ rounds if the topology is known. Both of the latter algorithms crucially rely on (network) coding.

\textbf{Network coding and coding gap:}
Network coding was first studied outside of the radio broadcast model by \citet{Ahlswede2006}.
For multiple-message broadcast, network coding algorithms by \citet{ghaffari2015randomized} and~\citet{Khabbazian2011} achieved a throughput of $\Omega(1/ \log n)$. A network coding gap of $\Omega(\log \log n)$ on certain topologies for the radio network model was demonstrated by~\citet{haeuplersodap1843} as was a $\Theta(1)$ worst case gap.

In wired undirected networks, \citet{Li2009} show the network coding gap to be at most $2$. For directed wired networks, \citet{Halperin2003} shows that the coding gap corresponds to the integrality gap of directed Steiner-tree LP. For the wired vertex-congest model, in which messages sent at the same time do not collide, a tight coding gap of $\Theta(\log n)$ was given by~\citet{Censor-Hillel2014}.

\textbf{Robust communication models:} The noisy broadcast model, introduced by \citet{ElGamal1984}, resembles our model as it assumes random errors. This model was mostly studied in the context of computing functions of the inputs of nodes~\cite{Gallager88, Newman04, GoyalKS08, kushilevitz1998computation}. However, this model assumes a complete communication network and single-bit transmissions. An extension of this line of work for random planar networks was also studied~\cite{KanoriaM2007,YingSD06,DuttaR08,DuttaKMR08}.
 Unlike our own model, in this model a node can receive a message from multiple neighbors in a single round.

Another model that captures unreliability in radio networks is the \emph{dual graph} model~\cite{KuhnLN11,Censor-HillelGKLN14,KuhnLNOR10,GhaffariHLN12,GhaffariLN13}. In this graph-based model, there is a static set of edges that form the network graph $G$ as well as an additional graph $G'$ which consists of edges that may or may not be present in each round, as chosen by an adversary. While an excellent way to capture an environment where some links are reliable and others are not, the dual graph model fails to model random noise.

\section{Preliminaries}
In this section, we review the classic radio network model, define our generalization of it, define throughput, explain how we quantify the coding gap and describe existing broadcast algorithms.

\subsection{The (Noisy) Radio Broadcast Model and $k$-Message Broadcast}
The classic radio network model consists of an undirected graph $G=(V, E)$ with $n$ nodes and diameter $D$. Nodes communicate in synchronized time steps (rounds) by either remaining silent (listening) or locally broadcasting a \textbf{packet} (each neighbor gets the same packet). A node $u$ receives a packet from a neighbor in round $r$ if and only if exactly one of its neighbors broadcasts in $r$ and $u$ remains silent. We term this model the \textbf{faultless} model.

We build on the classic model and introduce a \textbf{noisy radio network} model. Our model is the same as the classic radio broadcast model but with one of two modifications. In the \textbf{receiver faults} model, a node that is listening and has only one broadcasting neighbor receives noise with constant probability (independently of other nodes). In the \textbf{sender faults} model, a broadcasting node has a constant probability of transmitting noise (independently of other senders failing). We use $p \in [0, 1)$ to denote the fault probability and expose it in results where appropriate. In both variations of the model, we assume that if a node $u$ receives noise due to collisions, faults, or none of its neighbors broadcasting, the node does not mistake this noise for a legitimate packet from a neighbor.

The most commonly studied algorithmic problem in the radio network model is the $k$-message broadcast problem. This problem consists of a \textbf{source} $s \in V$ and $k$ \textbf{messages}. $(G,s)$ is often referred to as the \textbf{topology}. Each message is of size $O(\log (nk))$ and each packet is of size $O(\log (nk))$.\footnote{Papers more commonly define the message and packet size to be $O(\log n)$. However, this choice leads to tedious technical difficulties in the noisy setting. Our choice of message and packet size as $O(\log (nk))$ is reasonable for two reasons: (1) $k = \poly(n)$ is the most practical setting and (2) without using $O(\log k)$ bits one cannot even send a message identifier. This assumption will become particularly important for us when providing coding schedules, as it is what enables us to use Reed-Solomon coding to create packets that can be broadcast in a single round.} The source $s$ begins in the first round as the only node that knows the messages. The problem is solved once all nodes have all $k$ messages.

A \textbf{schedule} is defined by the assignment of a function $b_u^r(\cdot)$ to each node $u$ in each round $r$, which, intuitively, governs the behavior of node $u$ in round $r$. The schedule is static and does not change during rounds. The domain and co-domain of the function depend on the setting as follows. In a \textbf{routing} setting, the function $b_u^r$ takes no input (except the implicit arguments $u$ and $r$) and outputs either a \emph{stay silent token} or the index in $\{1, 2, ..., k\}$ of the message to broadcast. If the output of $b_u^r$ is the index of a message which $u$ has not received by round $r$, the node remains silent. In a \textbf{(network) coding} setting, the function $b_u^r$ takes as input the entire history of packets node $u$ received and outputs either a \emph{stay silent token} or an arbitrary packet of length $O(\log nk)$. We informally refer to an $\textbf{algorithm}$ as a procedure that outputs a schedule.

\subsection{Throughput}
The throughput of a given topology is roughly the maximum number of messages that can be transmitted per round as the number of messages, $k$, goes to infinity.
\begin{definition}[Topology Throughput]
The routing throughput of a topology $(G, s)$ is defined as
\begin{align*}
\tau(G, s) = \limsup_{k \rightarrow \infty} \frac{k}{\min_{\mcS_k}|\mcS_k|}
\end{align*}
where $|\mcS_k|$ is the number of rounds taken by the schedule $\mcS_k$. In the faultless setting, the minimum is taken over schedules that broadcast $k$ messages from $s$ to all nodes in $G$. In the faulty setting, the minimum is taken over schedules that succeed with probability at least $1 - \frac{1}{k}$.\footnote{In the faulty setting $\min_{\mcS_k}|\mcS_k|$ could be replaced by the minimum of the \emph{expected number of rounds} taken by any schedule that succeeds with probability 1. While our results can be adapted to this definition, this definition introduces various technical issues which we wish to avoid in order to simplify the presentation.}
\end{definition}
Given a set of schedules, we define the \textbf{throughput of the set of schedules} as above but where we minimize over the set (rather than all possible schedules).

In general, we show topology throughput lower bounds of $\tau$ in the noisy model by showing that, for any $\eps > 0$ and any $k$, there is a $k_0 > k$ such that there exists a schedule that broadcasts $k_0$ messages in $k_0X$ rounds, where $|\frac{1}{X} - \tau| \leq \eps$. Similarly, to show a throughput upper bound of $\tau$ in the noisy model we demonstrate that, for sufficiently large $k$, any schedule that broadcasts with a probability of failure of at most $1/k$ requires at least $k\tau^{-1}$ rounds.

\subsection{Comparing Routing and Network Coding}
There are a couple of natural ways to quantify the advantage that coding offers over routing. First, one might be interested in finding a fixed topology where the ratio of coding throughput to routing throughput is largest. Second, one might be interested in the worst case performance of coding compared to the worst case performance of routing. We formalize these two quantities as the \textit{shared topology gap} and the \textit{worst case topology gap} respectively. Let $\tau^R(G, s)$ and $\tau^\NC(G, s)$ be the throughputs of $(G,s)$ when using routing schedules and when using network coding schedules respectively.

\begin{definition}[Shared Topology Gap]
The shared topology gap is:
\begin{align*}
\mcG_s = \max_{G, s} \left[ \tau^\NC(G, s) / \tau^R(G, s) \right]
\end{align*}
\end{definition}

\begin{definition}[Worst Case Topology Gap]
\label{def:WCTGap}
The worst case topology gap is defined as
\begin{align*}
\mcG_w = \left [ \min_{G,s} \tau^\NC(G, s) \right ] / \left [ \min_{G,s} \tau^R(G, s) \right ]
\end{align*}
\end{definition}

We will also sometimes refer to the coding gap of a fixed network $(G,s)$ defined as $\tau^\NC(G, s) / \tau^R(G, s)$. Note that any lower bound on the coding gap of a network is a lower bound on $\mcG_s$. 
The two quantities above are related as shown by the following claim\shortOnly{, whose proof is deferred to \Cref{subsec:gapCompareProof}}.
\begin{restatable}{lemma}{sharedWorstCaseGapCompare}
\label{obs:sharedWorstCaseGapCompare}
The worst case topology gap is at most the shared topology gap, i.e. $\mcG_w \leq \mcG_s$.
\end{restatable}
\gdef\ProofsharedWorstCaseGapCompare{
  \begin{proof}
    Let $(G',s')$ be the minimizing assignment to $\min_{G, s} [\tau^R(G, s)]$.
    Thus,
    \begin{align*}
      \mcG_w & = \left[ \min_{G, s}  \tau^\NC(G, s) \right] / \left[ \min_{G, s}  \tau^R(G, s) \right] \\
             & \leq \tau^\NC(G', s') / \left[ \min_{G, s} \tau^R(G, s) \right] \\
             & = \tau^\NC(G', s') / \tau^R(G', s') \\
             & \leq\max_{G, s} [\tau^\NC(G, s) / \tau^R(G, s)] = \mcG_s.
    \end{align*}
  \end{proof}
}

\fullOnly{\ProofsharedWorstCaseGapCompare}
\shortOnly{}

Interestingly, we later prove that this observation holds with equality in the receiver faults setting but not in the sender faults setting.

\subsection{Broadcast Algorithms for Faultless Radio Networks}
In this section, for completeness, we present the broadcast algorithms for the faultless setting whose performance in the noisy setting are addressed in Section~\ref{sec:RobustBroadcast}.

\subsubsection{Decay}\label{sec:prereq:decay}
Here we describe the classic Decay algorithm~\cite{bar1992time} for broadcasting \textbf{a single message} from the source $s$ to every other node. For each round, we define the set of \textbf{informed nodes} as all the nodes that received the message up to that round.

\textbf{Algorithm:} Divide the rounds into phases of $O(\log n)$ rounds. During the $i^{th}$ round of each phase, where  $i \le O(\log n)$, each \textbf{informed} node broadcasts the message independently with probability $2^{-i}$. A simple calculation yields the following.

\begin{lemma}
  \label{prop:decay-spread-one-phase}
  If a node $v$ has an informed neighbor at the start of the phase, it becomes informed by the end of the phase with constant probability.
\end{lemma}
\begin{proof}
  If the number of informed neighbors that $v$ has in the $i^{th}$ round of a phase is in $[2^i, 2^{i+1})$, a simple calculation shows that it becomes informed by the end of the round with constant probability. Since the number of informed neighbors is nondecreasing, there is always a round in a phase in which the above condition holds.
\end{proof}

The round complexity of Decay then follows \cite{bar1992time}.
\begin{lemma}[\citet{bar1992time}] \label{lem:decayRuntime}
In the faultless setting, Decay spreads a single message in $O(D \log n + \log n(\log n + \log \frac{1}{\delta}))$ rounds with a probability of failure of at most $\delta$.
\end{lemma}

\begin{proof}
Fix a path $s = u_0, u_1, ..., u_l = v$ from the source $s$ to any node $v$ (the length $l$ of the path is at most the diameter $D$). At round $t$, let $\phi$ be the largest $i$ such that $u_i$ knows the message (initially, $\phi = 0$). After one phase of $O(\log n)$ rounds $\phi$ either remains the same or increases by 1 with constant probability, by \Cref{prop:decay-spread-one-phase}. Hence, after $O(D + \log n + \log \frac{1}{\delta})$ phases, the probability of failure can be bounded via a Chernoff bound:
\begin{align*}
  Pr[\Phi < l] < \exp\left(-\Omega\left(\log n + \log \frac{1}{\delta}\right)\right) .
\end{align*}
Applying a union bound over all $n$ nodes gives that the failure probability is at most $n \cdot \exp(-\Omega(\log n + \log \frac{1}{\delta})) < \exp(-\Omega(\log \frac{1}{\delta})) < \delta$.

\end{proof}

\subsubsection{FASTBC}\label{sec:prereq:fastbc}
We next describe an optimal algorithm for single-message broadcast when all the nodes \textbf{know the topology} beforehand, given by~\citet{gkasieniec2007faster}, which we refer to as the FASTBC algorithm. FASTBC works by first decomposing the graph into a \textbf{gathering-broadcasting spanning tree} (GBST) and then utilizing this structure to broadcast a message in $D + O(\log^2 n)$ rounds.

We need the following notion in order to describe the algorithm. A \textbf{ranked breadth-first search (BFS) tree} $\mcT$ for a graph $G$ is a BFS tree rooted at a source node $s$ where every node in the tree is assigned an integral rank. The ranks are assigned inductively, as follows. Every leaf node has rank $1$. A non-leaf node $v$ with a maximum child rank of $r$ is assigned a rank of $r$ if exactly one child of $v$ has rank $r$. Otherwise, if two or more children have rank $r$, then $v$ is assigned a rank of $r+1$.
Additionally, we define the level of node $v$ in a BFS tree as the distance from $s$ to $v$. The maximum rank can be bounded as follows.

\begin{lemma}[\citet{GaberCentralized}]\label{lem:rankBound}
The largest rank $r_{\max}$ in a ranked BFS tree of size $n$ is no greater than $\lceil \log_2 n \rceil$.
\end{lemma}

\textbf{GBST:} A ranked BFS tree $\mcT$ of a graph $G$ is a \textbf{gathering-broadcasting spanning tree} (GBST) if and only if no two distinct nodes on the same level and of the same rank $r$ have two distinct $\mcT$-parents both with rank $r$. See \Cref{fig:rankedBFSVSGBST} for a comparison of a ranked BFS tree that is not a GBST with a ranked BFS tree that is a GBST. Note that we assume nodes agree on a common GBST before the start of the algorithm because of the known topology assumption.

\begin{figure}
\centering
\subfigure[A ranked BFS Tree that is not a GBST.]{\label{subfig:rankedBFS}\includegraphics[scale=.2]{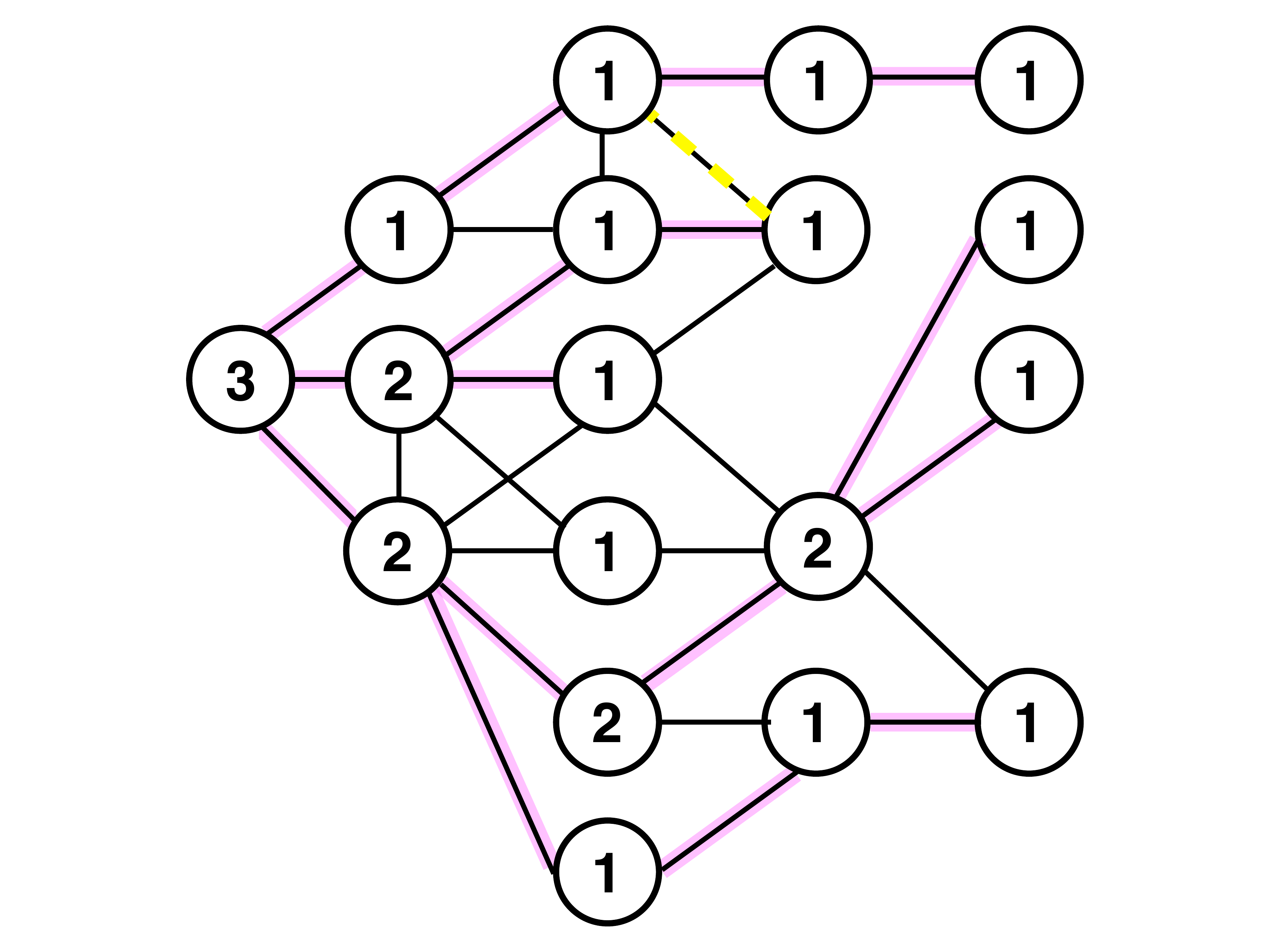}}
\subfigure[A ranked BFS Tree that is a GBST.]{\label{subfig:GBST}\includegraphics[scale=.2]{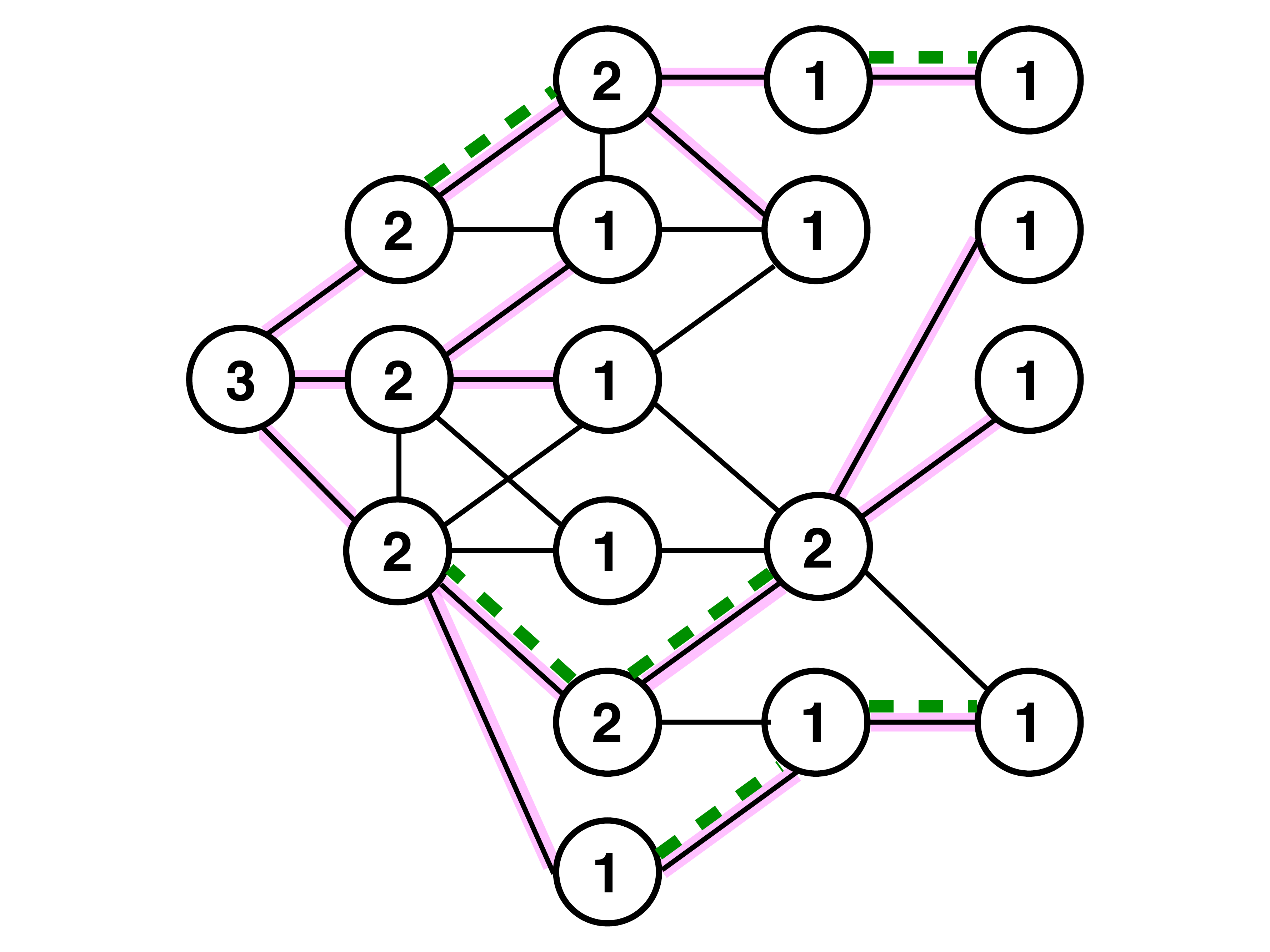}}
\caption{Black lines indicate edges in $G$; pink lines indicated edges in $\mcT$; the dashed yellow line indicates the edge which makes \Cref{subfig:rankedBFS} not a GBST; dashed green lines indicate a fast node and the corresponding child in the GBST. Example taken from~\citet{ghaffari2015randomized}.}
\label{fig:rankedBFSVSGBST}
\end{figure}

\textbf{Fast nodes in a GBST:} We call node $u$ \textbf{fast} if one of its GBST-children has the same rank as $u$, and that tree-edge is called a \textbf{fast edge}. Any GBST-path from the source to a node has nonincreasing ranks. Hence it is composed of at most $O(\log n)$ \textbf{fast stretches} of consecutive fast nodes of the same rank connected by fast edges.

\textbf{The FASTBC Algorithm \cite{gkasieniec2007faster}:}
At a high level, FASTBC divides the rounds into odd and even numbered rounds. On odd-numbered rounds (called \textbf{slow transmission rounds}) it performs a standard Decay algorithm that pushes a message from a higher ranked node to a lower ranked one. During even-numbered rounds (called \textbf{fast transmission rounds}) it pushes a message along fast stretches.

More formally, during a fast transmission round $2t$, only \textbf{fast} nodes on level $l$ and rank $r$ broadcast if $t \equiv l - 6r \pmod{6r_{\max}}$. During a slow transmission round $2t+1$, all nodes broadcast with probability $2^{-(t \mod O(\log n))}$.

Intuitively, messages on a fast stretch ride a \textbf{wave} from the start of the stretch to its tail. More formally, when a messages gets transmitted along a fast stretch for the first time, it never gets interrupted until it reaches the tail. This non-interference along fast stretches is ensured by the properties of the GBST. The round complexity of FASTBC is given as follows.

\begin{lemma}[\citet{gkasieniec2007faster}]
In the faultless setting, FASTBC spreads a single message in $D + O(\log n (\log n + \log \frac{1}{\delta}))$ rounds with a probability of failure of at most $\delta$.
\end{lemma}
\begin{proof}
Fix any path from the source $s$ to a node $v$. The path can be decomposed into at most $r_{\max} = O(\log n)$ fast stretches and $O(\log n)$ non-fast edges. By \Cref{lem:decayRuntime}, the message successfully traverses all of the non-fast edges in $O(\log n ( \log n + \log \frac{1}{\delta} ) )$ rounds with a probability of at least $1 - \frac{\delta}{n}$.

Next, we analyze the fast stretches during fast transmission rounds. Let the length of the $i^{th}$ fast stretch be $D_i$. The time for a message to reach the tail from the start of the fast stretch is $D_i + O(\log n)$ (with probability $1$) since it has to wait at most $6r_{\max} = O(\log n)$ rounds for the first transmission and it is not interrupted until it reaches the tail of the stretch.

  This non-interference is ensured because fast nodes of different ranks that transmit during the same round must be at least 6 levels apart (and since a GBST is a BFS tree, they do not interfere). Nodes of the same rank that transmit in the same round will not interfere because of the GBST construction property.

  Putting together the behaviors of the fast stretches and non-fast edges, we get that a message is transmitted along a path from $s$ to $v$ in $\sum_i \left (D_i + O(\log n) \right ) + O(\log n (\log n + \log \frac{1}{\delta})) = D + O(\log n (\log n + \log \frac{1}{\delta}))$ rounds with a probability of at least $1 - \frac{\delta}{n}$. Applying the union bound over all nodes completes the proof.
\end{proof}

\section{Robust Broadcast Algorithms}
\label{sec:RobustBroadcast}
In this section we describe how we adapt algorithms from the faultless setting to the sender or receiver faults setting in order to obtain robust broadcast algorithms.

\subsection{Robust Algorithms for Single-Message Broadcast}
Decay ~\cite{bar1992time} -- which spreads a message in $O(D\log n + \log^2 n)$ rounds with high probability\footnote{We denote ``with high probability'' when the probability of success is at least $1 - \frac{1}{n}$.} in the faultless setting -- is a classic faultless broadcasting algorithm. \fullOnly{We now prove that the Decay algorithm works as-is with faults.}\shortOnly{In \Cref{subsec:RobustDecay}, we prove that the Decay algorithm works as-is with faults.}

\begin{restatable}{lemma}{RobustDecay}
\label{RobustDecay}
With sender or receiver faults with fault probability $p$, the Decay algorithm spreads a single message in $O \left(\frac{\log n}{1 - p} (D + \log n + \log \frac{1}{\delta}) \right)$ rounds with a probability of failure of at most $\delta$.
\end{restatable}
\gdef\ProofRobustDecay{
\begin{proof}
Fix a path $s = u_0, u_1, ..., u_l = v$ from the source $s$ to any node $v$ (the length $l$ of the path is at most the diameter $D$). At round $t$, let $\phi$ be the largest $i$ such that $u_i$ knows the message (initially, $\phi = 0$). After one phase of $O(\log n)$ rounds $\phi$ either remains the same or increases by 1 with probability $c(1 - p)$ for a particular constant $c > 0$ by an analogue of \Cref{prop:decay-spread-one-phase} which appears in \Cref{sec:prereq:decay}. Hence, after $O(\frac{1}{1-p}(D + \log n + \log \frac{1}{\delta}))$ phases, the probability of failure can be bounded via a Chernoff bound:
\begin{align*}
  Pr[\Phi < l] < \exp\left(-\Omega\left(\log n + \log \frac{1}{\delta}\right)\right) .
\end{align*}
Applying a union bound over all $n$ nodes gives that the failure probability is at most $n \cdot \exp(-\Omega(\log n + \log \frac{1}{\delta})) < \exp(-\Omega(\log \frac{1}{\delta})) < \delta$.
\end{proof}
}

\fullOnly{\ProofRobustDecay}
\shortOnly{}

While the Decay algorithm is very simple and does not require nodes to know the topology in advance, the number of rounds it takes depends super-linearly on the diameter. However, if the topology is known in advance, it is possible to achieve linear dependence on the diameter in the faultless setting, as shown by the FASTBC algorithm~\cite{gkasieniec2007faster}. FASTBC succeeds with high probability in $O(D + \log^2 n)$ rounds. 

However, in contrast to the Decay algorithm, the performance of FASTBC deteriorates in the faulty setting. In particular, FASTBC uses a fragile round counting mechanism to synchronize message passing along its fast stretches which does not translate well to the faulty setting. The following result quantifies this performance deterioration.

\begin{restatable}{lemma}{fastbcBadFaults}\label{prop:fastbcBadFaults}
FASTBC takes $\Theta\left(\frac{p}{1-p} D \log n + \frac{1}{1-p}D \right)$ rounds in expectation to broadcast a message along a single path of length $D$.
\end{restatable}
\gdef\ProofFastbcBadFaults{
\begin{proof}
  Consider a path of length $D$ with a message on its left endpoint. Let $f(x)$ be the expected number of rounds to traverse the path of length $x$ with FASTBC under the condition that the leftmost node broadcasts in the next round. With probability $1-p$ the broadcast is successful, leading to a subproblem with length $x-1$. Alternatively, with probability $p$, it can be unsuccessful -- leading to a waiting time of $\Theta(\log n)$ until the leftmost node broadcasts again and without reducing the path length $x$. The resulting recurrence is $f(x) = 1 + (1-p) f(x-1) + p [f(x) + \Theta(log n)]$ which can be rewritten as $f(x) = f(x-1) + \frac{1}{1-p} + \frac{p}{1-p} \Theta(\log n)$, leading to $f(x) = \Theta(\frac{p}{1-p} x \log n + \frac{1}{1-p}x)$.
\end{proof}
}
\fullOnly{\ProofFastbcBadFaults}
\shortOnly{}

One simple way to construct a robust broadcast algorithm is to perform FASTBC but repeat each round $\Theta(\log n)$ times for a total of $O(D\log n + \log^{O(1)} n)$ rounds. This approach works because the probability that any message gets dropped is at most $\frac{1}{n^{\Omega(1)}}$, which allows one to apply a union bound over the entire algorithm (of length at most $\poly(n)$). However, this approach loses the linear dependence on $D$ and therefore performs no better than Decay. A better approach, inspired by \Cref{prop:fastbcBadFaults}, would be to repeat each message $\Theta(\log \log n)$ times, giving $p = \frac{1}{\log^{\Omega(1)} n}$ and a $O(D \log \log n + \poly\log n)$-round algorithm.

We further refine this approach in the following \textbf{Robust FASTBC} algorithm, which gives a linear dependence on $D$ in the noisy setting. 
We first provide the high-level idea of Robust FASTBC to give intuition and then provide a more formal, terse definition.

\textbf{The High-Level Idea of Robust FASTBC:} As in FASTBC, a GBST is constructed from the source node $s$. Consider any GBST-path from $s$ to another node $v$ and note that it has at most $r_{\max} = O(\log n)$ fast stretches (consecutive fast nodes of equal rank) because the ranks on the path are non-increasing.

The algorithm works by partitioning the rounds into odd and even-numbered rounds. During odd-numbered rounds, a standard Decay algorithm is performed on all nodes. These rounds are meant to push the message from one fast stretch to the next.

During even-numbered rounds, a different procedure makes progress along fast stretches. First, partition the nodes of each fast stretch into blocks of size $S := \Theta(\log \log n)$.\footnote{All the blocks have size $\Theta(\log\log n)$, except possibly the last one.} We define the procedure of \textbf{broadcasting on a block} in the following way: at round $2t$ (only even numbered rounds can be performing this procedure) a node with level $l$ broadcasts if it is in the block and $t \equiv l \pmod{3}$. The procedure continues on for $c \cdot S = \Theta(\log \log n)$ rounds for some sufficiently large constant $c$. We work modulo $3$ to prevent collisions between nodes at consecutive levels since we are working on a BFS tree (and the GBST properties prevents interferences from nodes at the same level). Note that the probability that a message that is in a broadcasting block in the beginning fails to exit the block is at most $\frac{1}{\log^{c'} n}$ for a constant $c'$ which we can set as large as needed (by increasing the round multiplier $c$).

Now imagine contracting the nodes in a block into one \textbf{supernode}. A broadcast on this supernode corresponds to the block-broadcast procedure described in the last paragraph and \textbf{superrounds} on this graph correspond to $\Theta(\log \log n)$ rounds in the original graph. We define ranks, levels and fast supernodes in the same way as in the original graph (in particular, the entire supernode gets the same level and rank).

The algorithm on the contracted graph is as follows: at round $t$, a fast supernode with level $l$ and rank $r$ broadcasts if $t \equiv l - 6r \pmod{6r_{\max}}$. In other words, a \textbf{wave} propagates the message from one end of the fast stretch to another end in a number of superrounds that is linear with respect to the length of the fast stretch. Broadcasts at the same superround of nodes of different ranks are displaced by at least 6 levels, hence they do not interfere with each other. Two consecutive waves on the same block are spread $6r_{\max} = O(\log n)$ superrounds apart.

\textbf{Formal Robust FASTBC Algorithm:} As in FASTBC, a gathering-broadcasting spanning tree is constructed from the source node. Let the round number of the algorithm be $t$. When $t$ is odd, each node will perform a standard Decay step (a node broadcasts with probability $2^{-(t-1)/2 \mod O(\log n)}$). Let $S = \Theta(\log \log n)$ and $c$ be a sufficiently large constant. At even-numbered rounds $t$, a node in the fast set with level $l$ and rank $r$ will broadcast if $\lfloor \frac{l}{S} \rfloor - 6r \equiv \lfloor \frac{t/2}{cS} \rfloor \pmod{6r_{\max}}$ and $l \equiv t \pmod{3}$.

The following is our main result for this section.
\begin{restatable}{theorem}{RobustFASTBC}
\label{RobustFASTBC}
Robust FASTBC spreads a single message in $O(D + \log n\log\log n(\log n + \log \frac{1}{\delta}))$ rounds with a probability of failure of at most $\delta$ if sender or receiver faults occur with probability $p$.
\end{restatable}
\gdef\ProofRobustFASTBC{
\begin{proof}
  Fix a GBST path $P = (s = u_0, u_1, u_1, ..., u_l = v)$ from the source $s$ to a node $v$. Partition the edges on the path into fast stretches (consecutive edges connecting two nodes of the same rank) and non-fast edges. There can be at most $O(\log n)$ non-fast edges interconnecting $O(\log n)$ fast stretches.

  Assuming a message is on a non-fast edge, during the next $\Theta(\log n)$ rounds it is transmitted along that edge with constant probability. Given that there are only $O(\log n)$ such edges, a Chernoff bound gives us that after $O(\log n(\log n + \log \frac{1}{\delta}))$ such rounds (where a message is waiting to be transmitted along a non-fast edge), the message is transmitted along all the non-fast edges on $P$ with probability at least $1 - \frac{\delta}{2n}$.

  Next, we turn to counting the number of rounds that a message spends on fast stretches (during even-numbered rounds). Note that from the design of the algorithm and the properties of the GBST no two broadcasting nodes ever interfere with each other (hence the only failures come from constant probability faults).

  Call a fast node from the path $P$ a \textbf{barrier} if its level is divisible by $S$ and call a message \textbf{active} if it is on a fast stretch and the node it is currently at is broadcasting. Note that a message that enters a fast stretch has to wait $r_{\max} cS = O(\log n \log\log n)$ rounds until it becomes active. Once it is active, consider its behavior during the next $cS = O(\log\log n)$ rounds. The message can either exit the stretch, remain active (meaning it reaches the next barrier) or become inactive (by failing $(c-1)S$ out of $cS$ transmissions). The probability of becoming inactive is at most $\frac{1}{\log^3 n}$ by Chernoff with an appropriate constant $c$. Every time a message becomes inactive, it waits $O(\log n \log\log n)$ rounds before it becomes active.

  Let $d_1, d_2, ..., d_q$ (note that $q \le O(\log n)$) be the lengths of the fast stretches in the path $P$. When a message is active, it traverses the paths in at most $\sum_{i=1}^q \lceil d_i/S \rceil \cdot cS = \sum_{i=1}^q O(d_i + 1) = O(D + \log n)$ rounds. 
  The number of rounds it takes for a message to become active is at most
  \begin{align*}
    & \left (q + \frac{T}{cS} Pr[\text{msg inactive in $cS$ rounds}] \right ) \cdot \\
    & \quad \quad \cdot O(\log n \log\log n) \\
    \le\ & O(\log^2 n\log\log n) + \frac{T}{\Theta(\log \log n)} \frac{1}{\log^3 n} O(\log n\log\log n) \\
    =\ & O(\log^2 n\log\log n) + \Theta \left (\frac{T}{\log^2 n} \right ),
  \end{align*}
  where $T$ is the length of the entire protocol. The $q$ term comes from becoming active each time a message enters a fast stretch. The $\frac{T}{cS}$ accounts for the possibility of a message becoming inactive in between barriers.  A Chernoff bound proves that if $T = \Theta(D + \log n \log\log n(\log n + \log \frac{1}{\delta}))$, the message gets passed along the path with a probability of at least $1 - \frac{\delta }{2n}$.

  Putting together the behavior during the odd-numbered and even-numbered rounds and applying a union bound over all $n$ nodes gives that the protocol forwards the message from the source to all other nodes in the claimed number of rounds with probability at least $1 - \delta$.
\end{proof}
}
\fullOnly{\ProofRobustFASTBC}
\shortOnly{}
\subsection{Robust Algorithms for Multi-Message Broadcast}
A pleasant feature of single-message broadcasting algorithms that are robust to sender failures is that they can be used in a black-box manner to transmit $k$ messages with random linear network coding~\cite{haeupler2011analyzing}, provided some minor technical conditions are satisfied. For instance, a node cannot change its behavior based on whether it receives a message or not. However, all of our algorithms can be made to satisfy these conditions using methods similar to~\cite{ghaffari2015randomized}. We state the results that can be achieved and refer the reader to~\citet{ghaffari2015randomized} and \citet{haeupler2011analyzing} for details.

\begin{lemma}\label{prop:DecayTPut}
Decay with random linear network coding can broadcast $k$ messages in $O(D\log n + k\log n + \log^2 n)$ rounds if sender or receiver faults occur with constant probability. It follows that any topology has a coding throughput of $\Omega\left(\frac{1}{\log n}\right)$.
\end{lemma}

\begin{lemma}
Robust FASTBC with random linear network coding can broadcast $k$ messages in $O(D + k\log n\log\log n + \log^2 n\log\log n)$ rounds if sender or receiver faults occur with constant probability. It follows that any topology has a coding throughput of $\Omega\left(\frac{1}{\log n \log\log n}\right)$.
\end{lemma}

We leave as an open problem the existence of an algorithm that is robust to sender and receiver faults and can broadcast $k$ messages in $O(D + k\log n + \poly \log(n))$ -- this would be optimal up to additive $\poly \log$ factors.
\section{Throughput Gaps with Noisy Broadcast}

We now study the gaps between (network) coding and routing in the noisy radio network model. In the faultless setting, by \Cref{prop:DecayTPut}, coding can send $k$ messages in $O(D\log n + k\log n + \log^2 n)$ rounds. Moreover, in the faultless setting, no routing scheme with $k$ polynomial in $n$ is known to send $k$ messages in fewer than $\Omega \left(k \log ^2 n \right)$ rounds. The apparent disparity of what is achievable with coding versus routing demands a formal explanation.

Previous work of \cite{haeuplersodap1843} attempts to formalize this gap. However, although this work shows a shared topology gap of $\Omega(\log \log n)$ in the faultless setting, it also shows the counterintuitive result of a worst case topology gap of $\Theta(1)$.

Our new model and results give a more satisfactory explanation. We formally show in what sense the high throughput routing schemes of \cite{haeuplersodap1843} are not robust; namely, they cease to be efficient for random receiver failures. Moreover, in the noisy radio network model we prove that coding is indeed necessary for high throughput broadcasting by exhibiting a worst case topology gap of $\Theta(\log n)$ and a shared topology gap of $\Omega(\log n)$ in the receiver fault setting. All of our routing lower bounds, and therefore all of our gaps, are particularly strong as we prove them in a setting where routing is allowed to be \textbf{adaptive}, as we later define.

Additionally, we show that with coding or adaptive routing and sender faults, very little differs from the faultless setting in terms of throughput: every coding (resp. routing) schedule that assumes no faults can be transformed into a coding (resp. adaptive routing) schedule in the sender faults setting with the same throughput up to constants. These transformations allow us to conclude the following: (1) interestingly, the worst case topology gap is highly sensitive to whether sender or receiver faults are examined -- we show it to be $\Theta(1)$ for sender faults with adaptive routing; (2) the $\Omega(\log \log n)$ shared topology gap of the faultless setting also exists in the sender fault setting.

We now define an adaptive routing schedule wherein all nodes are allowed to adapt to all faults.
\begin{definition}[Adaptive Routing Schedule]
An adaptive routing schedule is a sequence of functions $b_r(\cdot)$ -- one each every round $r$ -- that takes as input \begin{enumRomanHor}\item the entire topology $(G, s)$ and \item all tuples $(u, i)$ where node $u$ received message $m_i$ in some round $r' < r$
\end{enumRomanHor}. Each $b_r(\cdot)$ outputs a sequence of length $n$ directing each node to either remain silent or to broadcast a message they know.
\end{definition}
\noindent In practice, a distributed routing algorithm might receive some feedback as to when faults occur but it will certainly not receive as much as our adaptive routing schedules. Thus, our definition is sufficiently strong for proving meaningful gaps. Moreover, if routing is non-adaptive, a $\Theta(\log k)$ gap can be proved on a single-link topology as shown in \Cref{sec:single-link-gap}. However, with adaptive routing the gap is $\Theta(1)$ on the single-link topology. Thus, adaptivity can considerably improve routing throughput, further motivating adaptive routing as a strong model for proving gaps.

Lastly, we use Reed-Solomon coding~\cite{WickerRSCodes} as a black box for our coding schedules throughout this section. Given $k$ input packets, Reed-Solomon coding constructs $\poly(nk)$ coded packets such that any $k$ of the coded packets is sufficient to reconstruct the original $k$ packets.

\subsection{Gaps for Receiver Faults with Adaptive Routing}
In this section we study the coding gaps with receiver faults and adaptive routing.

\subsubsection{Star Topology: $\Theta(\log n)$ Gap for Receiver Faults with Adaptive Routing}
We first show a $\Theta(\log n)$ gap in the receiver fault setting on a star topology. A \textbf{star topology} consists of a node $s$ and $n$ other adjacent nodes.\footnote{We use $n$ instead of $n-1$ other nodes to simplify calculations.} \fullOnly{}\shortOnly{We prove the following in \Cref{subsec:StarLowerRoute}.}

\begin{restatable}{lemma}{StarLowerRoute}
\label{StarLowerRoute}
The adaptive routing throughput of the star topology with receiver faults is $\Theta(1/\log n)$.
\end{restatable}
\gdef\ProofStarLowerRoute{
\begin{proof}
For concreteness assume that $p=\frac{1}{2}$.
We prove a throughput of $\Omega \left(\frac{1}{\log n}\right)$ by providing the following schedule. The schedule works by broadcasting message $m_1$ from $s$ until it is received by all nodes, then broadcasting $m_2$ and so forth up to message $m_k$. The schedule stops after $13k + k \log n$ rounds. It remains to show that the schedule fails with a probability of at most $1/k$.

Let $X_i$ be the random variable that stands for the number of rounds required until message $m_i$ is received and let $X = \sum_i X_i$. Moreover, define $Y_i = X_i - \log n - 1$ and $Y = \sum_i Y_i$. It is not hard to show that for all $i$, $\E[X_i] \leq \log n + 1$, hence $Y_i$ roughly counts the number of rounds in which $X_i$ goes past its expectation. We now bound the tail distribution of the $Y_i$ variables, thereby bounding it for the $X_i$ variables and for $X$.

Let $Z_t^i$ be the random variable of the number of nodes that do not receive message $m_i$ after $s$ broadcasts $m_i$ for $t$ rounds. It holds that $\E[Z_t^i] = n(\frac{1}{2})^{t}$. The following shows that $\Pr(X_i \geq t) \leq  n \left(\frac{1}{2}\right)^{t}$.

\begin{align*}
n\left(\frac{1}{2}\right)^{t} &= \E[Z_t^i] = \sum_{j=1}^n j\Pr(Z_t^i = j)\\
&\geq \sum_{j=1}^n \Pr(Z_t^i = j)\\
&= \Pr(X_i \geq t),
\end{align*}

and thus, we have $\Pr(X_i \geq \log n + 1 + t) \leq (\frac{1}{2})^{t+1}$ and hence by definition of $Y_i$, we have $\Pr(Y_i \geq t) \leq (\frac{1}{2})^t$.

We now use Chernoff bounds for geometric random variables (see \Cref{thm:geomChernoffBound}) to bound the tail distributions of $Y$ and $X$. For every $1\leq i \leq k$, let $H_i$ be a geometric random variable with a probability of success of $\frac{1}{2}$ and let $H = \sum_i H_i$. Note that $\E[H] = 2k$ and recall that $\Pr(H_i \geq t) = (\frac{1}{2})^t$. Since $\Pr(Y_i \geq t) \leq (\frac{1}{2})^t = \Pr(H_i \geq t)$, we also have $\Pr(Y \geq t) \leq \Pr(H \geq t)$ and therefore applying Chernoff bounds for geometric random variables gives that for any $\delta > 0$,

\begin{align*}
\Pr(Y \geq (1+\delta)\E[H]) &\leq \Pr(H \geq (1+\delta)\E[H])\\
&\leq \exp\left(-\frac{\delta^2(k-1)}{2(1+\delta)}\right).
\end{align*}

Letting $\delta = 5$ gives that $\Pr(Y \geq 12k) \leq \frac{1}{k}$. Lastly, we note that $\Pr(Y \geq 12k) = \Pr(X \geq 13k + k \log n)$. Thus, our schedule sends all $k$ messages in $13k + k \log n$ rounds with a probability of failure no greater than $1/k$.

To finish the proof, we need to show that our schedule achieves a throughput of $\Omega(1/\log n)$. To do so we must show that for any $\eps$ and for any $k'$, there exist a sufficiently large $k_0 > k$ such that $\left|\frac{k_0}{13k_0 + k_0\log n} - \frac{1}{\log n} \right| \leq \eps$. For sufficiently large $n$ this is clearly satisfied.




We now prove an upper bound of $O(\frac{1}{\log n})$ on throughput by proving that $\Omega(k \log n)$ rounds are necessary to broadcast $k$ messages with a probability of failure of at most $1/k$. Because every node other than $s$ is connected only to $s$, we know that any successful schedule routes $k$ messages from $s$ to each of the $n$ nodes directly. Broadcasting from any node other than $s$ clearly does not help the schedule and so, without loss of generality, we assume that only $s$ broadcasts. Similarly, it clearly does not help a schedule to broadcast a message after it is received by all nodes and so, without loss of generality, the schedule only broadcasts a message if some node has still not received that message. Therefore, without loss of generality, an adaptive routing schedule on a star can be thought of as broadcasting each message from $s$ until it is received by all nodes (potentially a given message is not always broadcast in contiguous rounds).

As in the proof of the lower bound, let $X_i$ stand for the number of rounds in which $m_i$ is broadcast and let $X = \sum_{i}X_i$.

We now show that $\E[X_i] \geq (\log n)/4$. Suppose, for the sake of contradiction, that $\E[X_i] = (\log n)/4$. It follows by Markov's inequality that it is possible to send $m_i$ by broadcasting it for $(\log n)/2$ rounds with probability of failure at most $1/2$. However, the probability that a given node does not receive $m_i$ in $(\log n)/2$ rounds is $(\frac{1}{2})^{\log n / 2} = 1/\sqrt{n}$ and so the probability that all nodes receive the message is $\left(1-\frac{1}{\sqrt{n}} \right)^n \leq \exp(-\sqrt{n})$, by the inequality $1-x \leq e^{-x}$. Thus, broadcasting for $(\log n)/2$ rounds succeeds with probability at most $\exp(-\sqrt{n})$, which contradicts that it is possible to broadcast for $(\log n)/2$ rounds with a probability of failure at most $1/2$. Thus, $\E[X_i] \geq (\log n)/4$.

%
%
%
%

We now show that $(k \log n )/ 4$ rounds are insufficient to broadcast $k$ messages with a failure probability of at most $1/k$. Suppose, for the sake of contradiction, that $(k \log n )/ 4$ rounds are sufficient to broadcast $k$ messages with a probability of failure no greater than $1/k$. Let $S_i$ stand for the random variable that is 1 if $X_i \leq (\log n)/16$ and 0 otherwise and let $S = \sum_i S_i$. Every message is broadcast on average $(\log n)/16$ times and so by Markov's inequality there are at least $k/2$ messages each of which is sent fewer than $(\log n)/8$ times. Since our schedule succeeds with probability $1-1/k$, it follows that with probability at least $1-1/k$ there are at least $k/2$ messages that are successfully sent in fewer than $(\log n)/16$ rounds. Otherwise stated, $ \Pr(S \geq k/2) \geq 1- \frac{1}{k}$. In what follows, we contradict this fact.

Since $\E[X_i] \leq (\log n)/4$, we have $\E[S_i] = \Pr(X_i \leq (\log n) / 16) \leq \Pr(X_i \leq \E[X_i](1-.75))$. By Chernoff bounds, $\Pr(X_i \leq \E[X_i](1-.75)) \leq \exp(-\Omega(E[X_i])) \leq \exp(-\Omega(\log n)) = 1/n^c$, for appropriate $c$. Thus, $\E[S] \leq k/n^c$. We now use this bound on $\E[S]$ to bound the probability that $S$ is large. It holds that $\Pr(S \geq k/2) \leq \Pr(S \geq \E[S](1+\frac{n^c-2}{2}))$ and so by another Chernoff bound we have $\Pr(S \geq k/2) \leq \exp(-\Omega(\E[S])) \leq \exp(-\Omega(n^c))$.

Therefore, for sufficiently large $k$, $n$ and appropriate $c$, this contradicts the above fact that $ \Pr(S \geq k/2) \geq 1- \frac{1}{k}$. Thus, we conclude that for large $k$ and $n$, $(k \log n )/ 4$ rounds are insufficient to broadcast $k$ messages with a failure probability of at most $1/k$ and so the throughput is $O(1/\log n)$.
\end{proof}
}

\fullOnly{\ProofStarLowerRoute}
\shortOnly{}

\begin{restatable}{lemma}{StarCode}
\label{StarCode}
The coding throughput of the star topology with receiver faults is $\Theta(1)$.
\end{restatable}
\gdef\ProofStarCode{
\begin{proof}
A throughput of $O(1)$ trivially follows because $\Omega(k)$ rounds are strictly necessary.

We show a throughput of $\Omega(1)$ by providing a schedule that uses Reed-Solomon coding to send $k$ messages in $O(k)$ rounds. Using Reed-Solomon coding it is possible to generate and send $100k + 100\log n$ packets such that reception of any $k$ of the packets by the receiver suffices for the receiver to reconstruct the original $k$ messages. Our schedule creates these packets and then broadcasts them. By Chernoff bounds, the probability that a given node does not receive $k$ packets is $\exp (-c(k + \log n)) \leq 1/(n2^k)^c$ for some positive constant $c \ge 2$. By a union bound over the nodes, we have that the probability that some node does not receive $k$ packets is at most $1/(n2^k)^{c-1} < 1/k$. For sufficiently large $k$, we have $100k + 100\log n = O(k)$. Lastly, to show a throughput of $\Omega(1)$ we note that for any $\eps > 0$ there exists a sufficiently large $k_0$ such that $|100 - \frac{k_0}{100k_0 + 100\log n}| \leq \eps$.
\end{proof}
}

\fullOnly{\ProofStarCode}
\shortOnly{\ProofStarCode}

\Cref{StarLowerRoute} and~\Cref{StarCode} give the claimed gap on the star topology.
\begin{theorem}
\label{cor:starGap}
In the receiver faults setting there is a $\Theta(\log n)$ coding gap for the star topology with adaptive routing. A shared topology gap of $\Omega(\log n)$ in the receiver faults setting follows.
\end{theorem}

\subsubsection{Worst Case Topology: $\Theta(\log n)$ Gap for Receiver Faults with Adaptive Routing}
Having shown a shared topology gap of $\Omega(\log n)$ in the receiver fault setting we now show a worst case topology gap of $\Theta(\log n)$ with receiver faults. Recall that we define the worst case topology gap as $\min_{G, s} \tau^\NC(G, s) / \min_{G, s} \tau^R(G, s)$. To prove the $\Theta(\log n)$ gap, we show that with receiver faults and adaptive routing, $\min_{G, s} \tau^\NC(G, s) = \Theta(1/\log n)$ and $\min_{G, s} \tau^R(G, s) = \Theta(1/\log^2 n)$.

We first describe the topology, $\WCT$, that has minimal throughput for both coding and adaptive routing with receiver faults. The $\WCT$ topology is based on a construction of Ghaffari et al.~\cite{ghaffari2013bound}; Ghaffari et al.~\cite{ghaffari2013bound} demonstrate the existence of a bipartite network of radius 2 with $\tilde{\Theta}(\sqrt{n})$ nodes. These $\tilde{\Theta}(\sqrt{n})$ nodes consist of a source node, $\tilde{\Theta}(\sqrt{n})$ receiver nodes, $\Theta(\sqrt{n})$ sender nodes. Every receiver node is connected to a subset of the sender nodes by a probabilistic construction. See \Cref{fig:throughputLowerBound} for a sketch of the network. We duplicate each of the $\tilde{\Theta}(\sqrt{n})$ receiver nodes to induce $\tilde{\Theta}(\sqrt{n})$ star-like topologies which we term clusters. More formally,  $\WCT$ is as follows.

\newpage

\vspace{0.1cm}
\begin{mdframed}[roundcorner=4pt, backgroundcolor=white]
\textbf{Worst case topology ($\WCT$) for receiver faults:}
We begin with the construction of Ghaffari et al.~\cite{ghaffari2013bound}. Instead of each receiver node, $r$, we construct a \textbf{cluster} of $\tilde{\Theta}(\sqrt{n})$ nodes. Each node in the cluster that replaces receiver $r$ has an edge to sender node $s$ if and only if $\{s, r\}$ is an edge of the original network of Ghaffari et al.~\cite{ghaffari2013bound}. We use the symbol $\WCT$ for the resulting network. See \Cref{fig:throughputLowerBoundMod} for an illustration of $\WCT$.
\end{mdframed}
\vspace{0.1cm}

Since nodes in a cluster are connected to the same sender nodes, a node in a cluster is sent a packet without collision if and only if all other nodes in the cluster are sent the same packet without collision. Moreover, Ghaffari et al.~\cite{ghaffari2013bound} prove that at most $O(1 / \log n)$ of all receiver nodes receive a packet without collision in any one round in the original topology, so we conclude the following.
\begin{lemma}\label{prop:numClustersHit}
At most $O(1/ \log n)$ of clusters of $\WCT$ receive a packet without collision per round.
\end{lemma}

We now use this topology to prove the following claim \fullOnly{}\shortOnly{in \Cref{subsec:AdapRoutWCT}}.
\begin{figure}
\centering
\subfigure[The lower bound network from \cite{ghaffari2013bound}]{\label{fig:throughputLowerBound}\includegraphics[scale=.2]{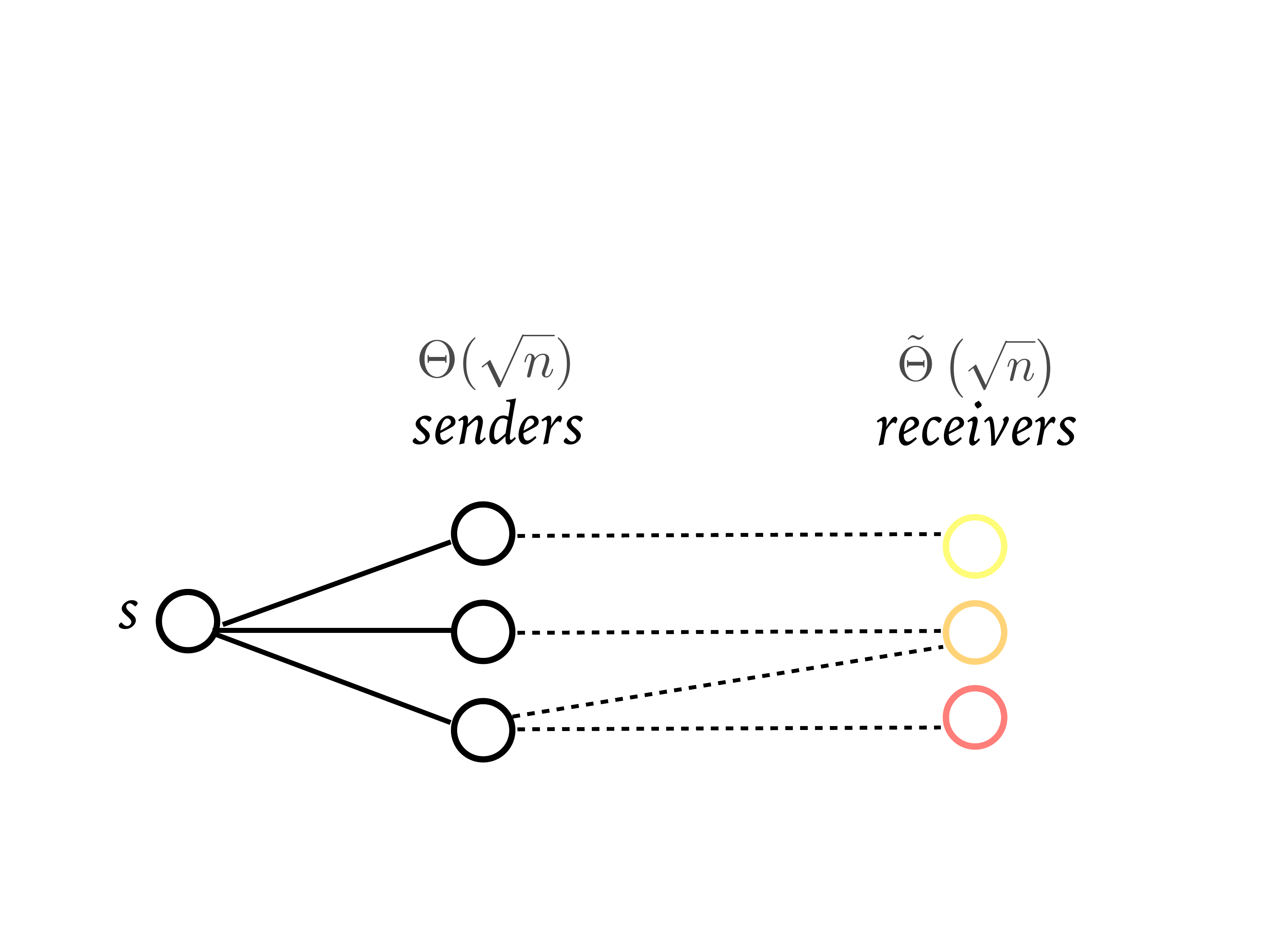}}
\subfigure[The modified lower bound network, $\WCT$]{\label{fig:throughputLowerBoundMod}\includegraphics[scale=.2]{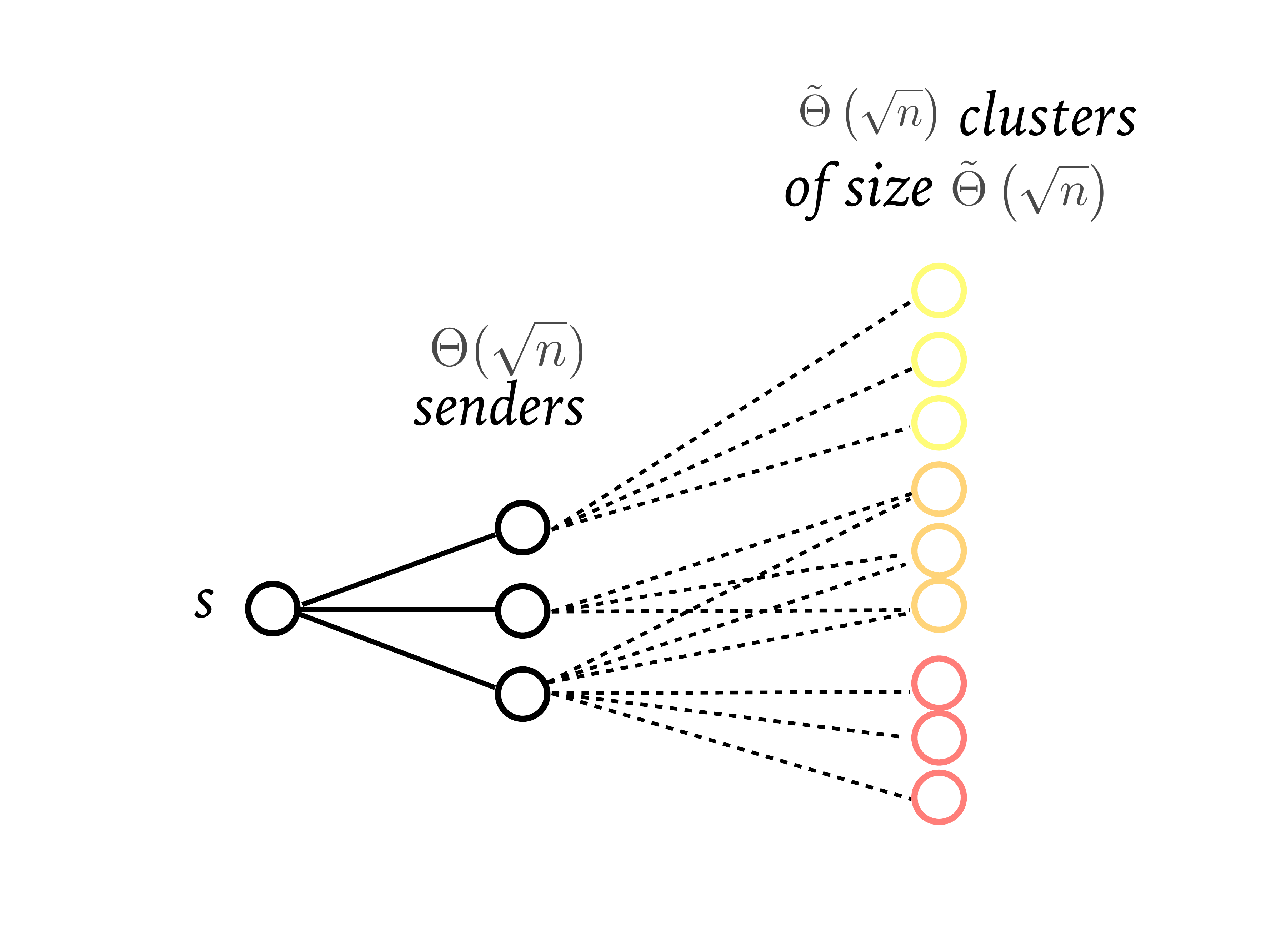}}
\caption{We modify the network from \cite{ghaffari2013bound} to a worst case topology in the receiver fault setting; dotted lines refer to edges from the probabilistic construction, solid lines refer to definite edges and clusters are colored according to the receiver node from which they are derived.}
\label{fig:throughputLowerBounds}
\end{figure}

\begin{restatable}{lemma}{AdapRoutWCT}
\label{AdapRoutWCT}
  Adaptive routing on $\WCT$ has a throughput of $O(1 / \log^2 n)$ in the receiver fault setting.
\end{restatable}
\gdef\ProofAdapRoutWCT{
\begin{proof}
Consider the above described $\WCT$. Since in each round every node in a cluster is sent the same message without collision or no message, we can interpret each as a star of $\tilde{\Theta}(\sqrt{n})$ nodes. By \Cref{StarLowerRoute}, each cluster must receive a packet without collision at least $\Omega(k \log (\sqrt{n})) = \Omega(k \log n)$ times, in order to receive $k$ messages with a probability of failure of at most $1/k$. Thus, in order for every node to receive a message such that the probability of failure is at most $1/k$, every cluster must receive a message without collision at least $\Omega(k \log n)$ times. By \Cref{prop:numClustersHit}, $O\left(\frac{1}{\log n}\right)$ clusters are sent a packet without collision each round and so at least $\Omega((\log n) (k\log n)) = \Omega ( k \log ^2 n)$ rounds are necessary to have a failure probability of at most $1/k$. Thus, the throughput is at most $O(1 / \log^2 n)$.
\end{proof}
}

\fullOnly{\ProofAdapRoutWCT}
\shortOnly{}

We next move to the $\Omega(1 / \log^2 n)$ adaptive routing throughput in the receiver fault setting. To prove this possibility result we first prove an intermediate result for bipartite networks. \fullOnly{}\shortOnly{The proofs for the bipartite case and the general case appear in \Cref{subsec:RoutingBipartite} and \Cref{subsec:AdapRoutingRecFault}, respectively.}

\begin{restatable}{lemma}{RoutingBipartite}
\label{RoutingBipartite}
  Consider bipartite network $V = L \cup R$ where every node in $L$ knows the same $k$ messages. There exists an adaptive routing schedule of length $O(k \log^2 n)$ that broadcasts the $k$ messages to all the nodes in $R$ with probability at least $1 - \exp(-\Omega(k))$ in the receiver fault setting.
\end{restatable}
\gdef\ProofRoutingBipartite{
\begin{proof}
 We show that there exists an adaptive routing schedule, $\mcS_j$, that succeeds in broadcasting $k$ messages from all nodes in $L$ to all nodes in $R$ with a probability of failure of at most $\exp(-\Omega(k))$. By \Cref{RobustDecay}, Decay can route all $k$ messages to $R$ in $O(\log^2 n)$ rounds with a failure probability of at most $1/n$. Denote by $\mcS_{m_i}$ this schedule provided by Decay when it is used for message $m_i$. The schedule $\mcS_j$ runs the schedule $\mcS_{m_1}$ repeatedly until it succeeds, then it runs $\mcS_{m_2}$ until it succeeds and so forth until $\mcS_{m_{k}}$. However, $\mcS$ never runs more than $\frac{2k}{1-1/n}$ of these schedules in total.

 Let $C_i$ be the random variable that stands for the number of times $\mcS_{m_i}$ must be run and let $C = \sum_i C_i$. Note that $C_i$ is a geometric random variable with probability of success $1 - 1/n$ and that $\E[C] = k/(1-\frac{1}{n})$. By a Chernoff bound for geometric random variables (see \Cref{thm:geomChernoffBound}), we have that the probability that $\mcS$ does not succeed is

 \begin{align*}
 \Pr\left(C \geq \frac{(1+1)k}{1-1/n}\right) &\leq \exp \left(-\frac{1^2(k - 1)}{2(1+1)}\right)\\
 &= \exp(-\Omega(k)).
 \end{align*}

 Thus, $\mcS_j$ succeeds in sending $k$ messages in $\frac{2k}{1-1/n}(\log ^ 2 n) = \Omega(k \log ^ 2 n)$ rounds with a probability of failure of at most $\exp(-\Omega(k))$.
\end{proof}
}

\fullOnly{\ProofRoutingBipartite}
\shortOnly{}

\begin{restatable}{lemma}{AdapRoutingRecFault}
\label{AdapRoutingRecFault}
  Adaptive routing on any network has a throughput of $\Omega(1 / \log^2 n)$ with receiver faults.
\end{restatable}
\gdef\ProofAdapRoutingRecFault{
\begin{proof}
  We now prove that the worst case adaptive routing throughput is $\Omega(1 / \log^2n)$ by providing an adaptive routing schedule. Roughly, our schedule works by \textbf{pipelining} schedules given by \Cref{RoutingBipartite}. We first note that any broadcast problem given a source $s$ and graph $G$ can be broken into a series of broadcast problems on bipartite graphs. In particular, it can be broken into the BFS layering of $G$ where the $i^{th}$ layer contains all nodes with distance exactly $i$ from $s$. Let $L_i$ be the set of nodes in the $i^{th}$ layer. Note that each pair of consecutive layers $L_i$, $L_{i+1}$ define a bipartite network. Next, we divide $k$ into $D$ batches of size $k' = \frac{k}{D}$ (assume, without loss of generality, that $D$ divides $k$), which we pipeline.

By \Cref{RoutingBipartite}, we can broadcast batch $j \in [D]$ of $k'$ messages in $O(k'\log^2 n)$ rounds with a failure probability of at most $\exp(-\Omega(k'))$. Let $\mcS_j$ be the schedule that broadcasts batch $j$ in this manner.


We now describe our schedule for broadcasting through the entire network. We divide rounds into meta-rounds of size $\Omega(k' \log^2n)$. The $k$-th layer runs $\mcS_{j}$ in meta-round $3j + k - 3$. In other words, we pipeline batches through the network using the above schedule for bipartite graphs, working layers of $3$ apart so we do not incur extra collisions. We run $4(D+\frac{k}{k'})$ meta-rounds.

We argue that this schedule achieves a throughput of $\Omega(1/\log^2 n)$. A fixed $\mcS_j$ fails in a meta-round with a probability of at most $\exp(-\Omega(k')) = \exp(-\Omega(k/D)) \leq 1/k^c$, for sufficiently large $k$. By a union bound over the diameter, the probability that any $\mcS_j$ fails in a meta-round is at most $1/k^{c-1}$. Moreover, we run $4(D+\frac{k}{k'})$ meta-rounds and so by another union bound over meta-rounds, an $\mcS_j$ fails in any meta-round with a probability of at most $1/k^{c-2} \leq 1/k$, for an appropriate $c$. Lastly, if every $\mcS_j$ in every meta-round succeeds, then every node receives every batch, meaning that every node receives all messages. Since each meta-round is of length $\frac{2k'}{1-1/n}\log^2n$ and we run $4(D+\frac{k}{k'})$ meta-rounds, we use a total of $4(D+\frac{k}{k'})(\frac{2k'}{1-1/n}\log^2n) \leq 24k\log^2n$ rounds to succeed with probability at least $1-1/k$, for $n\geq 2$.

Finally, to conclude a throughput of $\Omega(1/\log^2n)$ we need to show that for any $\eps > 0$ and $k$ there exists a $k_0$ such that $|\frac{k_0}{24k_0\log^2n} - \frac{1}{24\log^2}| \leq \eps$, which trivially holds.
\end{proof}
}

\fullOnly{\ProofAdapRoutingRecFault}
\shortOnly{}

The impossibility result of \Cref{AdapRoutWCT} and the possibility result of \Cref{AdapRoutingRecFault} yield the following.
\begin{lemma}\label{worstCaseRoutingTPut}
  The worst case adaptive routing throughput with receiver faults is $\Theta(1/\log^2 n)$, i.e. $\min_{G, s} \tau^R(G, s) = \Theta(1/\log^2 n)$.
\end{lemma}

Next, we show that coding in the receiver fault setting can always achieve a throughput of $\Omega(1 / \log n)$, and that no better bound exists. \fullOnly{}\shortOnly{The proof of this appears in \Cref{subsec:WCTCode}.}

\begin{restatable}{lemma}{WCTCode}
\label{WCTCode}
The worst case coding throughput in the receiver fault setting is $\Theta(1/\log n)$, i.e. $\min_{G, s} \tau^\NC(G, s) = \Theta(1/\log n)$.
\end{restatable}
\gdef\ProofWCTCode{
\begin{proof}
We first prove that the worst case coding throughput is $O(1/ \log n)$ with receiver faults. Consider the above worst case topology. In any given round at most $O(1/\log n)$ clusters receive a message. Each cluster forms a star and, similarly to $\Cref{StarCode}$, each star must be sent the message $\Theta(k)$ rounds, in order to decode all $k$ messages. Thus, $\Omega(k\log n)$ rounds are strictly necessary and in particular are  necessary to succeed with probability at least $1-1/k$. We conclude a coding throughput of $O(1/\log n)$ on this topology.

By \Cref{prop:DecayTPut}, the worst case coding throughput is $\Omega(1/ \log n)$ with receiver faults, which completes the proof.
\end{proof}
}
\fullOnly{\ProofWCTCode}
\shortOnly{}

By \Cref{worstCaseRoutingTPut} and \Cref{WCTCode}, we conclude our strong worst case topology gap.
\begin{theorem}
  The worst case topology gap is $\Theta(\log n)$ for receiver faults with adaptive routing.
\end{theorem}

\subsection{Transformations from the Faultless Setting to the Faulty Setting}
Having shown that there is a strong $\Theta(\log n)$ gap in the receiver fault setting, we now turn our attention to the sender fault setting. We begin by presenting explicit transformations of schedules from the faultless setting into schedules that are robust to faults. \fullOnly{}\shortOnly{The proofs of the following appear in \Cref{subsec:TPutRoutingTransformProof} and \Cref{subsec:TputForNoFaultToFaultCoding}, respectively.}

\begin{restatable}{lemma}{TputForNoFaultToFault}
\label{lem:TputForNoFaultToFault}
A set of routing schedules with throughput $\tau$ in the faultless setting can be transformed into a set of adaptive routing schedules with throughput $\tau(1-p)$ for the sender fault setting.
\end{restatable}
\gdef\ProofTputForNoFaultToFault{
  \begin{proof}
    We construct a set of routing schedules, $\mcS'$, such that for any $\eps'$ and any $k'$ there exists a $k_0' \geq k'$ such that $\left|\tau (1-p) - \frac{k_0'}{|\mcS_{k_0'}|} \right| \leq \eps'$, where $\mcS'_{k_0}$ is a schedule in $\mcS'$ that broadcasts $k_0'$ messages with probability of success at least $1-1/k_0'$ and $|\mcS_{k_0'}|$ is the number of rounds it uses. Fix $\eps'$ and $k'$. We now show how to construct  $\mcS'_{k_0'}$ that satisfies the above inequality.

    Suppose we have a set of schedules $\mcS$ that achieves a throughput of $\tau$ in the faultless setting. That is, for any $\eps > 0$ and any $k$ there exists sufficiently large $k_0 \geq k$, such that $\left|\frac{k_0}{|\mcS_{k_0} |} - \tau \right| \leq \eps$ for some schedule $\mcS_{k_0}$ in $\mcS$. It follows that $|\mcS_{k_0}| \leq k_0(\tau - \eps)^{-1}$. We use $\mcS_{k_0}$ to construct $\mcS'_{k_0}$, such that the schedule $\mcS'_{k_0}$ sends $k_0' = k_0x$ messages in $(\tau - \eps)^{-1}(1-p)^{-1}(1 + \eta)$ rounds with probability at least $1-1/k_0'$, for a value of $x$ to be chosen later and arbitrarily small $\eps$ and $\eta$. Notice that we can guarantee that $k_0' \geq k'$ by picking $k$ sufficiently large, since $k_0' = k_0x \geq kx$.

    The schedule $\mcS'_{k_0'}$ is constructed as follows. Each round of broadcast of $\mcS_{k_0}$ corresponds to a meta-round of $\mcS'_{k_0'}$ composed of $x(1-p)^{-1}(1 + \eta)$ rounds, for $\eta > 0$ chosen as small as desired. In each meta-round of $\mcS'_{k_0'}$, if a node broadcasts message $m_i$ in the corresponding round of $\mcS_{k_0}$, the node now broadcasts $m_{i1}$ until it succeeds, then it broadcasts $m_{i2}$ until it succeeds and so forth up to $m_{ix}$, until it reaches $x(1-p)^{-1}(1 + \eta)$ rounds. If it succeeds in sending all $x$ messages before the end of the meta-round, the node remains silent for the remainder of the meta-round.

    In each meta-round, in expectation each message requires $(1-p)^{-1}$ rounds to be sent successfully and so all $x$ messages for a meta-round require $(1-p)^{-1}x$ rounds in expectation. By Chernoff bounds, a node fails to send all $x$ messages in a given round with a probability of at most $\exp\left( -\Omega(x \eta^2)\right)$. Thus, letting $x = \Omega((\log (n k_0 \tau^{-1}))/\eta^2)$, the schedule $\mcS'_{k_0'}$ fails with a probability of at most $1/(nk_0 \tau^{-1})^c$, for some positive constant $c$. A union bound over all nodes shows that the probability that any node fails in a meta-round is at most $1/(nk'\tau^{-1})^{c-1}$. Moreover, by another union bound over meta-rounds, the probability that any node fails in any one of the $k_0 (\tau - \eps)^{-1}$ meta-rounds is at most $1/(nk_0 \tau^{-1})^{c-2}$, which is at most $1/k_0'$, for sufficiently large $c$. Thus, the schedule $\mcS_{k_0'}'$ succeeds in sending $k_0'$ messages in $k_0'(\tau - \eps)(1-p)(1 + \eta)^{-1}$ rounds with s probability of success of at least $1-1/k_0'$.

    Lastly, we need to show that $\left|\tau (1-p) - \frac{k_0'}{|\mcS_{k_0'}|}\right| \leq \eps'$. However, notice that $\left|\tau (1-p) - \frac{k_0'}{|\mcS_{k_0'}|} \right| = \left|\tau (1-p) - (\tau - \eps)(1-p)(1 + \eta)^{-1} \right|$, where $\eps$ and $\eta$ are arbitrarily small by our choice and hence the inequality clearly holds.
  \end{proof}
}
\fullOnly{\ProofTputForNoFaultToFault}
\shortOnly{}

\begin{restatable}{lemma}{TputForNoFaultToFaultCoding}
\label{TputForNoFaultToFaultCoding}
  A set of coding schedules from the faultless setting with throughput $\tau$ can be transformed into a set of coding schedules with throughput $\tau (1-p)$ in the sender or receiver fault setting.
\end{restatable}
\gdef\ProofTputForNoFaultToFaultCoding{
\begin{proof}
We construct a set of routing schedules, $\mcS'$, such that for any $\eps'$ and any $k'$ there exists a $k_0' \geq k'$ such that $\left|\tau (1-p) - \frac{k}{|\mcS_{k_0'}|}\right| \leq \eps'$, where $\mcS'_{k_0}$ is a schedule in $\mcS'$ that broadcasts $k_0'$ messages with probability of success at least $1-1/k_0$ and $|\mcS_{k_0'}|$ is the number of rounds it uses. Fix $\eps'$ and $k'$. We now show how to construct a schedule $\mcS'_{k_0'}$ that satisfies the above inequality.

Suppose we have a set of schedules $\mcS$ that achieves a throughput of $\tau$ in the faultless setting. That is, for any $\eps > 0$ and any $k$, there exists sufficiently large $k_0 \geq k$, such that $\left|\frac{k_0}{|\mcS_{k_0} |} - \tau \right| \leq \eps$ for some schedule $\mcS_{k_0}$ in $\mcS$. It follows that $|\mcS_{k_0}| \leq k_0(\tau - \eps)^{-1}$.
We use $\mcS_{k_0}$ to construct $\mcS'_{k_0}$, such that $\mcS'_{k_0}$ sends $k_0' = k_0x$ messages in $k_0'(\tau - \eps)^{-1}(1-p)^{-1}(1-\eta)^{-1}$ rounds with probability at least $1-1/k_0'$, for a value of $x$ to be chosen later and arbitrarily small $\eps$ and $\eta$. Notice that we can guarantee that $k_0' \geq k'$ by picking $k$ sufficiently large, since $k_0' = k_0x \geq kx$.

We construct $\mcS'_{k_0'}$ as follows. A message $m_i$ of $\mcS_{k_0}$ corresponds to $x$ messages of $\mcS_{k_0}'$, denoted by $m_{i1}, m_{i2}, \ldots, m_{ix}$. Each round of broadcast of $\mcS_{k_0}$ corresponds to a meta-round of $\mcS_{k_0'}'$ composed of $x(1-p)^{-1}(1-\eta)^{-1}$ rounds, for arbitrarily small $\eta > 0$. Let $f_u^r(m_1, \ldots, m_k)$ stand for the coded packet that $u$ broadcasts in round $r$ of $\mcS_{k_0}$. 
In $\mcS'_{k_0'}$, $u$ broadcasts in the corresponding meta-round as follows: $u$ computes, $F = f_u^r(m_{11}, \ldots, m_{k1}), \ldots, \allowbreak f_u^r(m_{1x}, \ldots, m_{k_0x})$; $u$ then uses Reed-Solomon coding on $F$ to create $x(1-p)^{-1}(1-\eta)^{-1}$ packets such that reception of any $x$ of these packets is sufficient to reconstruct all elements of $F$; $u$ broadcasts these $x(1-p)^{-1}(1-\eta)^{-1}$ packets over the course of the $x(1-p)^{-1}(1-\eta)^{-1}$ rounds of its meta-round.

We now argue that to show that $\mcS'_{k_0'}$ succeeds with probability at least $1- 1/k_0'$, it suffices to show that with probability at least $1-1/k_0'$ every node receives at least $x$ packets in every meta-round corresponding to a round of $\mcS_{k_0}$ in which the node received a broadcasted packet. Node $u$ is able to broadcast in a meta-round corresponding to round $r$ of $\mcS_{k_0}$ if it is able to construct the $x(1-p)^{-1}(1-\eta)^{-1}$ packets. Node $u$ is capable of constructing these packets if it is able to compute all elements of $F$, which it can do if it is able to compute $F' = f_{v_{r'}}^{r'}(m_{11}, \ldots, m_{k1}), \ldots, f_{v_{r'}}^{r'}(m_{1x}, \ldots, m_{k_0x})$ for any $r'<r$ where $r'$ is a round of $\mcS_{k_0}$ in which $u$ receives $f_{v_{r'}}^{r'}(m_{1}, \ldots, m_{k})$ from a neighbor $v_{r'}$. Node $u$ is able to compute $F'$ if in the meta-round corresponding to $r'$, $u$ receives at least $x$ packets. Thus, to show that $\mcS'_{k_0'}$ succeeds with probability at least $1- 1/k_0'$, it suffices to show that with probability at least $1-1/k_0'$ all nodes receive at least $x$ packets in any meta-round corresponding to a round of $\mcS_{k_0}$ in which they receive a coded packet. We show this as follows.

Consider a meta-round where a node is supposed to receive at least $x$ packets. The expected number of packets received by the node over the course of the meta-round is $x(1-\eta)^{-1}$. By a Chernoff bound, the probability that a node receives fewer than $x$ packets over the course of its meta-round is no greater than $\exp(-\Omega(x(1-\eta)))$. Letting $x = \Omega(\log (nk_0(\tau - \eps)^{-1})/(1-\eta))$, the probability that a node receives fewer than $x$ packets in a meta-round is at most $1/(nk_0(\tau - \eps)^{-1})^c$, for a positive constant $c$. By a union bound over the nodes, the probability that any node does not receive at least $x$ packets over the course of a meta-round is no greater than $1/(nk_0(\tau^{-1} - \eps))^{c-1}$. By a union bound over meta-rounds it follows that the probability that any node in any meta-round does not receive at least $x$ packets is at most $1/(nk_0(\tau - \eps)^{-1})^{c-2}$, which is at most $1/k_0'$, for sufficiently large $c$. Thus, the schedule $\mcS'$ succeeds in sending $k_0'$ messages in $k_0'(\tau - \eps)^{-1}(1-p)^{-1}(1-\eta)^{-1}$ rounds with a probability of failure of at most $1/k_0'$.

Lastly, we need to show that $\left|\tau (1-p) - \frac{k_0'}{|\mcS_{k_0'}|}\right| \leq \eps'$. However, notice that $\left|\tau (1-p) - \frac{k_0'}{|\mcS_{k_0'}|} \right| = \left|\tau (1-p) - \tau (1-p)(\tau - \eps)(1-\eta) \right|$, where $\eps$ and $\eta$ are arbitrarily small by our choice and hence the inequality clearly holds.
\end{proof}

}

\fullOnly{\ProofTputForNoFaultToFaultCoding}
\shortOnly{}

\subsection{Gaps for Sender Faults with Adaptive Routing}
We now use our transformations to derive a shared topology gap of $\Omega(\log \log n)$  and a worst case topology gap of $\Theta(1)$ in the \emph{sender fault setting} if routing is adaptive. This is a stark departure from the $\Theta(\log n)$ worst case topology gap of the \emph{receiver fault setting}. \fullOnly{}\shortOnly{The proof of the following appears in \Cref{subsec:SenderTopGap}.}




\begin{restatable}{theorem}{SenderTopGap}
\label{SenderTopGap}
The shared topology gap is $\Omega(\log \log n)$ with senders faults and adaptive routing.
\end{restatable}
\gdef\ProofSenderTopGap{
\begin{proof}
In~\cite{haeuplersodap1843}, a network that has a routing throughput of $O(1/\log \log n)$ in the faultless setting is given. A throughput upper bound of $O(1 / \log \log n)$ in the faultless setting clearly implies a throughput upper bound of $O(1 / \log \log n)$ in the sender fault setting. The work of~\cite{haeuplersodap1843} also provides a set of coding schedules that achieves a throughput of $\Omega(1)$ on the same network in the faultless setting. By \Cref{TputForNoFaultToFaultCoding}, we conclude a coding throughput of $\Omega(1)$ in sender fault setting. Thus, we conclude an $\Omega(\log \log n)$ gap on this topology and so a shared topology gap of $\Omega(\log \log n)$.
\end{proof}
}

\fullOnly{\ProofSenderTopGap}
\shortOnly{}

\begin{restatable}{theorem}{SenderWCTGap}
\label{SenderWCTGap}
The worst case topology gap is $\Theta(1)$ in the sender fault setting with adaptive routing.
\end{restatable}
\gdef\ProofSenderWCTGap{
\begin{proof}
Recall that the worst case topology gap is $\min_{G, s} \tau^\NC(G, s) / \min_{G, s} \tau^R(G, s)$.
We first show that $\min_{G, s} \tau^\NC(G, s) = \Theta(1/\log n)$. It is shown in~\cite{haeuplersodap1843} that there exist topologies where coding requires a throughput of $O(1/\log n)$ in the faultless setting. Note that a throughput upper bound in the faultless setting clearly implies one for the sender fault setting. Additionally, coding can achieve a throughput of $\Omega(1/ \log n)$ in the sender fault setting, by \Cref{prop:DecayTPut}.

We now show $\min_{G, s} \tau^R(G, s) = \Theta(1/\log n)$. It is shown in~\cite{haeuplersodap1843} that there exist topologies where routing requires a throughput of $O(1/\log n)$ in the faultless setting, and again we note that a throughput upper bound in the faultless setting clearly implies one for the sender fault setting. The work of~\cite{haeuplersodap1843} also shows that in the faultless setting there exists a set of routing schedules with a throughput of $\Omega(1 / \log n)$. We conclude, by \Cref{lem:TputForNoFaultToFault} it is possible to achieve a routing throughput of $\Omega(1 / \log n)$ in the sender fault setting.

Thus, $\min_{G, s} \tau^\NC(G, s) = \Theta(1/\log n)$ and $\min_{G, s} \tau^R(G, s) = \Theta(1/\log n)$ and so we conclude that with sender faults and adaptive routing the worst case topology gap is $\Theta(1)$.
\end{proof}
}

\fullOnly{\ProofSenderWCTGap}
\shortOnly{\ProofSenderWCTGap}


\newpage

\bibliography{./refs}

\appendix

\section{Single-link Topology Gaps}\label{sec:single-link-gap}
We show here that with non-adaptive routing on the trivial single-link topology, which consists of exactly two nodes $s, t$ connected by an edge, a shared topology gap of $\Omega(\log k)$ is easily attainable with sender or receiver faults.

\begin{lemma}
\label{obs:routeSingleLinkNonAdap}
The routing throughput on the single-link topology with constant sender or receiver fault probability but without adaptive schedules is $\Theta\left(\frac{1}{\log k} \right)$.
\end{lemma}
\begin{proof}
For concreteness, we assume that $p = 1/2$ in this proof. We first show that the non-adaptive routing throughput on the single-link topology is $O(1 / \log k)$. We do so by proving that $\Omega(k\log k)$ rounds are necessary for broadcasting $k$ messages with a probability of failure of at most $1/k$, for sufficiently large $k$.

Suppose, for the sake of contradiction, that $(k \log k)/4$ rounds are sufficient for sending $k$ messages with a probability of failure of at most $1/k$ by some schedule $\mcS$.

Since every message is sent $(k \log k)/4$ times on average, by Markov's inequality there must exist a subset of size $k/2$ for which no message in the subset is sent more than $(\log k)/2$ times. Call this subset $K'$.

Since $\mcS$ succeeds in sending all of the $k$ messages with probability at least $1-\frac{1}{k}$, it must succeed in sending all $k/2$ messages in $K'$ with probability at least $1-\frac{1}{k}$. Moreover, by construction of $K'$, it sends no message in $K'$ more than $(\log k)/2$ times.

We now show that no schedule can succeed in routing $k/2$ messages with probability $1-\frac{1}{k}$ if it sends no one of these $k/2$ messages more than $(\log k) /2$ times. If no message is sent more than $(\log k) /2$ times, the probability that a fixed message is not received is at least $(1/2)^{(\log k)/2} = 1/ \sqrt{k}$. Since each message succeeds independently, the probability that all $k/2$ messages are received successfully is $(1- 1/\sqrt{k})^{k/2}$ and so the probability that some message is not received is $1 - (1- 1/\sqrt{k})^{k/2}  \geq 1 - e^{-\sqrt{k}/2}$, by the inequality $e^{-x} \geq 1-x$. Therefore, we have that the probability that one of the $k/2$ messages is not received is at least $1 - e^{-\sqrt{k}/2}$, implying that the probability that all messages are received is no greater than $e^{-\sqrt{k}/2}$. Since $e^{-\sqrt{k}/2} \leq 1-1/k$ for sufficiently large $k$, we conclude that 
$\Omega(k\log k)$ rounds are necessary for broadcasting $k$ messages with a probability of failure of at most $1/k$, for sufficiently large $k$, implying that the throughput is $O(1/\log k)$.

We now show that a non-adaptive routing throughput of $\Omega(1/\log k)$ is achievable on the single-link topology. We do so by sending each message $100 \log k$ times from the source. The probability that a fixed message is not received is $\frac{1}{2}^{100\log {k}} = 1/k^{100}$. By a union bound the probability that any message is not received is $1/k^{99} \leq 1/k$. Thus, we broadcast $k$ messages in $100k \log k$ rounds with a probability of failure of at most $1/k$, showing a throughput of $\Omega(1/ \log k)$.\end{proof}

\begin{lemma}
\label{obs:codeSingleLink}
The coding throughput on the single-link topology with constant fault probability is $\Theta(1)$.
\end{lemma}
\begin{proof}
A throughput of $O(1)$ trivially follows from the fact that $\Omega(k)$ rounds are necessary to broadcast.


To show that a throughput of $\Omega(1)$ is achievable, we provide a schedule which uses Reed-Solomon codes. Using Reed-Solomon codes it is possible to generate and send $100k$ packets such that the reception of any $k$ packets by the receiver suffices for the receiver to reconstruct the original $k$ messages. The source $s$ generates and sends these $100k$ packets. By Chernoff bounds, the probability that fewer than $k$ messages are received is bounded from above by $\exp(-ck) \leq 1/k$, for some appropriate constant $c$. Thus, we succeed in broadcasting $k$ messages in $100k$ rounds with a failure probability of at most $1/k$, showing a throughput of $\Omega(1)$.
\end{proof}

Lemmas~\ref{obs:routeSingleLinkNonAdap} and~\ref{obs:codeSingleLink} give the following corollary.
\begin{lemma}
\label{cor:singLinkGap}
The single-link topology has a $\Theta(\log k)$ coding gap in the receiver/sender fault, non-adaptive routing setting, demonstrating an $\Omega(\log k)$ shared topology gap in this setting.
\end{lemma}

We now examine a stronger lower bound model - one with adaptive routing. The single-link topology coding gap disappears in this setting.
\begin{lemma}
\label{obs:routeSingleLinkAdap}
The adaptive routing throughput on the single-link topology with receiver/sender faults is $\Theta(1)$.
\end{lemma}
\begin{proof}
A throughput of $O(1)$ trivially follows because $\Omega(k)$ rounds are strictly necessary.

We prove a throughput of $\Omega(1)$ by providing the following adaptive routing schedule. The source sends each message until it is received, but for no more than $\frac{4k}{1-p}$ rounds in total. For every $1\leq i\leq k$, let $C_i$ be the random variable that stands for the number of rounds during which the $i$-th message $m_i$ is transmitted and let $C = \sum_i C_i$. Since $C_i$ is a geometric random variable with probability of success $1-p$, it holds that $\E[C] = \frac{k}{1-p}$. By the Chernoff bound for geometric random variables (see \Cref{thm:geomChernoffBound}), we have that the probability of failure is $\Pr(C \geq (1+3)\frac{k}{1-p}) \leq \exp(-\frac{9(k-1)}{2(1+3)})$, which is at most $1/k$ for sufficiently large $k$. Thus, we succeed in sending $k$ messages in $\frac{4k}{1-p}$ rounds with a probability of failure of at most $1/k$, showing a throughput of $\Omega(1)$.
\end{proof}

Lemmas~\ref{obs:codeSingleLink} and \ref{obs:routeSingleLinkAdap} give the following constant gap on the single-link topology.
\begin{lemma}
\label{cor:singLinkGapAdap}
For sender or receiver faults, there is a $\Theta(1)$ coding gap for the single-link topology if routing is done adaptively.
\end{lemma}

\shortOnly{
\section{Omitted Proofs}
\label{sec:omProofs}
}

\shortOnly{
\subsection{Proof of \Cref{obs:sharedWorstCaseGapCompare}}
\label{subsec:gapCompareProof}
\sharedWorstCaseGapCompare*
\ProofsharedWorstCaseGapCompare
}

\shortOnly{
\subsection{Proof of \Cref{RobustDecay}}
\label{subsec:RobustDecay}
\RobustDecay*
\ProofRobustDecay
}

\shortOnly{
\subsection{Proof of \Cref{prop:fastbcBadFaults}}
\label{subsec:fastbcBadFaults}
\fastbcBadFaults*
\ProofFastbcBadFaults
}

\shortOnly{
\subsection{Proof of \Cref{RobustFASTBC}}
\label{subsec:RobustFASTBC}
\RobustFASTBC*
\ProofRobustFASTBC
}

\shortOnly{
\subsection{Proof of \Cref{StarLowerRoute}}
\label{subsec:StarLowerRoute}
\StarLowerRoute*
\ProofStarLowerRoute
}

\shortOnly{
\subsection{Proof of \Cref{StarCode}}
\label{subsec:StarCode}
\StarCode*
\ProofStarCode
}

\shortOnly{
\subsection{Proof of \Cref{AdapRoutWCT}}
\label{subsec:AdapRoutWCT}
\AdapRoutWCT*
\ProofAdapRoutWCT
}

\shortOnly{
\subsection{Proof of \Cref{RoutingBipartite}}
\label{subsec:RoutingBipartite}
\RoutingBipartite*
\ProofRoutingBipartite
}

\shortOnly{
\subsection{Proof of \Cref{AdapRoutingRecFault}}
\label{subsec:AdapRoutingRecFault}
\AdapRoutingRecFault*
\ProofAdapRoutingRecFault
}

\shortOnly{
\subsection{Proof of \Cref{WCTCode}}
\label{subsec:WCTCode}
\WCTCode*
\ProofWCTCode
}

\shortOnly{
\subsection{Proof of \Cref{lem:TputForNoFaultToFault}}
\label{subsec:TPutRoutingTransformProof}
\TputForNoFaultToFault*
\ProofTputForNoFaultToFault
}

\shortOnly{
\subsection{Proof of \Cref{TputForNoFaultToFaultCoding}}
\label{subsec:TputForNoFaultToFaultCoding}
\TputForNoFaultToFaultCoding*
\ProofTputForNoFaultToFaultCoding
}

\shortOnly{
\subsection{Proof of \Cref{SenderTopGap}}
\label{subsec:SenderTopGap}
\SenderTopGap*
\ProofSenderTopGap
}

\shortOnly{
\subsection{Proof of \Cref{SenderWCTGap}}
\label{subsec:SenderWCTGap}
\SenderWCTGap*
\ProofSenderWCTGap
}

\section{Other Tools}
We use the following Chernoff bound for geometric random variables.
\begin{theorem}[\citet{doerr2011analyzing}]
\label{thm:geomChernoffBound}
Let $p \in (0,1)$. Let $X_1, \ldots , X_n$ be independent geometric random variables with $\Pr(X_i = j) = (1-p)^{j-1}p$ for all $j \in \mathbb{N}$ and let $X = \Sigma_{i=1}^n X_i$. Then for all $\delta > 0$,
\begin{align*}
\Pr(X \geq (1+\delta)\E[X]) \leq \exp\left(-\frac{\delta^2(n-1)}{2(1+\delta)}\right).
\end{align*}
\end{theorem}

%

\end{document}